\documentclass[12pt,fleqn]{article}
\usepackage{graphicx}
\usepackage{setspace}
\usepackage{amsmath}
\usepackage{amsthm}
\usepackage{amsfonts}
\usepackage{multirow}
\usepackage{booktabs}
\usepackage{amssymb}
\usepackage{authblk}
\usepackage{url}
\usepackage{undertilde}
\usepackage{subfig}
\RequirePackage[comma]{natbib}
\numberwithin{equation}{section}
\newcommand{\av}{\mbox{\boldmath$\alpha$}}

\newcommand{\bv}{\mbox{\boldmath$\beta$}}

\newcommand{\hbvg}{\mbox{\boldmath$\hat \beta$}}

\newcommand{\ev}{\mbox{\boldmath$e$}}
\newcommand{\epiv}{\mbox{\boldmath$\epsilon$}}

\newcommand{\Gv}{\mbox{\boldmath$G$}}
\newcommand{\Kv}{\mbox{\boldmath$K$}}

\newcommand{\Hv}{\mbox{\boldmath$H$}}

\newcommand{\gv}{\mbox{\boldmath$g$}}
\newcommand{\gav}{\mbox{\boldmath$\gamma$}}

\newcommand{\Qv}{\mbox{\boldmath$Q$}}

\newcommand{\Sigv}{\mbox{\boldmath$\Sigma$}}

\newcommand{\Mv}{\mbox{\boldmath$M$}}

\newcommand{\Xv}{\mbox{\boldmath$X$}}
\newcommand{\xv}{\mbox{\boldmath$x$}}

\newcommand{\Wv}{\mbox{\boldmath$W$}}
\newcommand{\uv}{\mbox{\boldmath$u$}}
\newcommand{\Vv}{\mbox{\boldmath$V$}}
\newcommand{\Psiv}{\mbox{\boldmath$\Psi$}}
\newcommand{\Phiv}{\mbox{\boldmath$\Phi$}}
\newcommand{\yv}{\mbox{\boldmath$y$}}

\newcommand{\Zv}{\mbox{\boldmath$Z$}}
\newcommand{\zv}{\mbox{\boldmath$z$}}
\newcommand{\wv}{\mbox{\boldmath$w$}}
\newcommand{\Iv}{\mbox{\boldmath$I$}}
\newcommand{\lv}{{\bf 1}}

\newcommand{\piv}{\mbox{\boldmath$\pi$}}

\newcommand{\ath}{{\it A. thaliana }}
\newcommand{\BF}{\ensuremath{{\rm BF}}}
\newcommand{\ABF}{\ensuremath{{\rm ABF}}}

\newtheorem{prop}{PROPOSITION}
\newtheorem{lemma}{LEMMA}

\title{{\bf Bayesian Model Comparison in Genetic Association Analysis: Linear Mixed Modeling and SNP Set Testing}}

\author
{Xiaoquan Wen \\
Department of Biostatistics, University of Michigan, Ann Arbor, USA}
\date{}
\begin{document}

\maketitle

\begin{abstract}
  {We consider the problems of hypothesis testing and model comparison under a flexible Bayesian linear regression model whose formulation is closely connected with the linear mixed effect model and the parametric models for SNP set analysis in genetic association studies. 
We derive a class of analytic approximate Bayes factors and illustrate their connections with a variety of frequentist test statistics, including the Wald statistic and the variance component score statistic. 
Taking advantage of Bayesian model averaging and hierarchical modeling, we demonstrate some distinct advantages and flexibilities in the approaches utilizing the derived Bayes factors in the context of genetic association studies. 
We demonstrate our proposed methods using real or simulated numerical examples in applications of single SNP association testing, multi-locus fine-mapping and SNP set association testing.}  
\end{abstract}

\section{Introduction}

In the past decades, genetic association studies have taken a prominent position in uncovering the role of genetic variants in disease etiology. 
Most recently, two related statistical approaches have become especially important in the analysis of genetic association data: the use of linear mixed models (LMM) to control for confounding factors and account for polygenic effects and the application of SNP set analysis for regions of (rare) genetic variants. 
As demonstrated by many authors \citep{Kang2010,Segura2012, Zhou2012, Zhou2013}, linear mixed models effectively thwart the identification of false positive associations caused by relatedness or population structures (e.g., cryptic relatedness) in the samples while at the same time increase the power of detecting genuine genetic association signals. 
SNP set testing \citep{Madsen2009, Wu2011, Lee2012} is emerging as a method of choice in detecting associations of rare genetic variants, which may be critical in explaining the phenomenon of ``missing heritability".     
Recent studies have also shown the necessity of jointly applying both approaches when analyzing the genetic association of rare variants to control for population stratification or using pedigree data.

Currently, the majority of the methodological work  employing LMM and/or SNP set analysis in genetic association studies has focused on reporting $p$-values for hypothesis testing. 
In this paper, we discuss a Bayesian alternative to address both topics within the model comparison framework in which hypothesis testing is regarded as a special case. 
We first show that both problems can be naturally formulated by a unified Bayesian parametric model, and we then derive a class of analytic approximate Bayes factors for use as our primary statistical device for model comparison. 
We establish the connections between the approximate Bayes factors and various commonly applied frequentist test statistics in a similar fashion, as reported by \cite{Wakefield2009, Wen2014a, Wen2014b}.

Despite its similarities in performance to the frequentist approaches in traditional hypothesis testing settings, the Bayesian approach exhibits great convenience and flexibility in dealing with complicated practical settings within and beyond hypothesis testing. 
One of the most significant advantages of the Bayesian comparison method is its acceptance of explicitly modeling various alternative scenarios (which are not necessarily nested) and the fluidity with which it combines the evidence from the data via Bayesian model averaging. 
Beyond single unit (i.e., either a SNP or a SNP set) association testing, we show that the Bayesian model comparison approach can be straightforwardly extended to a joint analysis of multiple association signals, especially when dealing with linkage disequilibrium (LD) among SNPs commonly present in the genetic data. 
We illustrate a highly-efficient multi-locus fine-mapping approach that is facilitated by our results based on approximate Bayes factors.

\section{Model and Notations}

We consider a general form of the linear mixed model,
\begin{equation}\label{lmm}
  \yv = \Xv \av + \Gv \bv + \uv  + \ev,~ \ev \sim {\rm N}(0, \tau^{-1} \Iv),
\end{equation}
where $\yv$ is an $n$-vector of quantitative response measurements, $\Xv$ is an $n \times q$  matrix of covariate variables to be controlled as {\em fixed effects} and their coefficients are encoded in the $q$-vector $\av$.
$\Gv$ is an $n \times p$  matrix of covariates whose effect, represented by the $p$-vector $\bv$, is of primary interest for inference. 
Finally, the $n$-vectors $\uv$ and $\ev$ represent the {\em random effects} and the i.i.d residual errors, respectively.
In the general LMM inference framework, the random effects vector, $\uv$, is assumed to be drawn from a multivariate normal (MVN) distribution, i.e.,
\begin{equation} \label{re.def}
  \uv \sim {\rm N}(\bf{0}, \lambda \tau^{-1} \Kv ),
\end{equation} 
where the $n \times n$ matrix $\Kv$ is assumed known (while the variance component parameter $\lambda$ is typically unknown).
In typical genetic applications,  $\Gv$ represents the genotypes of $p$ candidate SNPs, $\Xv$ includes intercept term and factors like age, sex that need to be controlled for, and $\uv$ usually represents the random effects due to cryptic genetic relatedness or population structure. The ultimate goal is to make inference of the genetic effect $\bv$. 

We now present a Bayesian counterpart of the LMM, the likelihood part of which is identical to (\ref{lmm}). 
From the Bayesian perspective, it is natural to regard the ``random effect" assumption (\ref{re.def}) as a standard MVN prior on $\uv$.   
For {\em controlled} ``fixed" effect coefficient $\av$,  we assume the MVN prior:
\begin{equation}
  \av \sim {\rm N}( \bf{0}, \Psiv), 
\end{equation}  
where  $\Psiv$ is a diagonal matrix. When performing inference, we take the limit   
$\Psiv^{-1} \to \bf{0}$,
which essentially assigns independent flat priors to each fixed effect coefficient.
A flat prior might be interpreted as an assumption that the {\it a priori} effects of $\av$ are extremely large. This assumption intuitively leads to a conservative inference on $\bv$. However, for variables that must be controlled for, such conservative assumptions are most likely welcome. 

We also assign an MVN prior for the parameter of interest, $\bv$, such that
\begin{equation}
  \bv \sim {\rm N}( \bf{0}, \Wv).
\end{equation}
The variance-covariance matrix $\Wv$ fully characterizes a distinct candidate model in our model comparison framework. The choice of $\Wv$ is context-dependent and has critical implications on the inference results. 
In practice, we recommend modeling the effect size on the unit-free scales of signal-noise ratios \citep{Wen2014a, Wen2014b} by assigning an MVN prior on the {\em standardized effect}, i.e., $\sqrt{\tau}\,\bv \sim {\rm N}(\bf{0}, \Phiv)$, which induces a prior variance matrix on the original scale of $\bv$ as $\Wv = \tau^{-1} \Phiv$. (Note that the prior on the random effect $\uv$ is formulated in the same scale.)

Finally, we assume a general joint prior distribution,  $p(\lambda, \tau)$, for the variance component parameters. 
As we will show later, the actual functional form of $p(\lambda, \tau)$ has little impact on our asymptotic approximations of Bayes factors.
To emphasize the connection with the frequentist linear mixed effect model, we will henceforth call the above Bayesian linear regression model the Bayesian linear mixed effect model (BLMM).

\section{Model Comparison in the BLMM}

We derive Bayes factors for the BLMM in order to perform Bayesian model comparisons. 
More specifically, we consider a space of candidate BLMMs that only differ in their specifications of $\Wv$. 
We denote $H_0$ as the trivial null model, in which $\bv \equiv 0$  (or equivalently $\Wv = \bf{0}$), and we define a null-based Bayes factor for an alternative model characterized by its prior variance on $\bv$ as
\begin{equation}
  \label{bf.def}
   \BF(\Wv) = \lim_{\Psiv^{-1} \to 0} \frac{P(\yv \mid \Xv,  \Gv, \Kv, \Wv)}{P(\yv \mid \Xv, \Gv, \Kv, H_0)} =  \lim_{\Psiv^{-1} \to 0}  \frac{P(\yv \mid \Xv, \Gv, \Kv, \Wv)}{P(\yv \mid \Xv, \Gv, \Kv, \Wv=\bf{0})}.
\end{equation}  


To present our results regarding the Bayes factors, we begin by introducing several necessary additional notations. 
We denote $\hat \av, \hat \bv, \hat \lambda$ and $\hat \tau$ as the MLEs of the full LMM model (\ref{lmm}) by treating $\bv$ as a fixed effect parameter. In addition, we denote $\hat \Vv = {\rm Var}(\hat \bv)$. 
Correspondingly, we use $\tilde \av, \tilde \tau$ and $\tilde \lambda$ to represent the MLEs of the null model, where $\bv$ is restricted to $0$. 
Furthermore, provided that parameter $\lambda$ is {\em known}, we note that $\hat  \bv$  can be analytically computed as a function of $\lambda$ (Appendix A.1), which we denote by $\hat \bv (\lambda)$. 
Accordingly, we use $\hat \Vv(\lambda, \tau)$ to represent the corresponding variance of  $\hat \bv(\lambda)$ (specifically, when $\lambda = \hat \lambda$, $\hat \bv (\hat \lambda) = \hat \bv$ and $\hat \Vv(\hat \lambda, \hat \tau) = \hat \Vv$). 
Finally, we consider a class of general estimators of $\lambda$, denoted by $\check \lambda$, for which a tuning parameter $\kappa \in [0,1]$ is built-in. 
The statistical details of this class of estimators are explained in Appendix A.2. Most importantly, it follows that $\check \lambda (\kappa = 0) = \tilde \lambda,$ and
$  \check \lambda (\kappa = 1) = \hat \lambda$ for the two extreme $\kappa$ values.
Deriving from $\check \lambda$, we establish a corresponding estimator of $\tau$, denoted by $\check \tau (\kappa)$, which can be analytically expressed in terms of $\check \lambda$ (see Appendix A.2) and also shares a similar property such that
$\check \tau (\kappa = 0) = \tilde \tau$, and $  \check \tau (\kappa = 1) = \hat \tau$.
Finally, we use the notations $\check \bv = \hat \bv (\check \lambda), \check \Vv = \hat \Vv(\check \lambda, \check \tau)$. In the case that $\Wv$ is specified as a function of $\lambda$ and/or  $\tau$, we denote $\check \Wv = \Wv(\check \lambda, \check \tau)$.
With these additional notations, we show that the desired Bayes factor can be approximated analytically. We summarize the main result in proposition 1, whose formal proof is given in Appendix A.2. 
\begin{prop}
Under the BLMM, the Bayes factor can be approximated by 
\begin{equation} \label{abf.exp}
   \ABF(\Wv, \kappa) = | \Iv + \check \Vv^{-1} \check \Wv |^{-\frac{1}{2}} \cdot \exp \left( \frac{1}{2} \check \bv' \check \Vv^{-1} \left[\check \Wv ( \Iv + \check \Vv^{-1} \check \Wv)^{-1} \right] \check \Vv^{-1} \check \bv \right).
\end{equation}
It follows that
 \begin{equation*}
 \BF(\Wv) = \ABF(\Wv, \kappa)\cdot \left(1+ O \left(\frac{1}{n}\right)\right), ~\mbox{ for any } \kappa \in [0,1].
 \end{equation*}  
\end{prop}
\noindent{\bf Remark 1.}  The approximate Bayes factors in the BLMM share the same functional form as the ABFs discussed in \cite{Wen2014a} and enjoy some of the computational properties discussed therein. In particular, the computation of the ABF is robust to the potential collinearity presented in the data matrix $\Gv$. Furthermore, $\Wv$ is allowed to be rank-deficient.  

\noindent{\bf Remark 2.} For single SNP analysis, i.e., $p=1$, both $\check \Wv$ and $\check \Vv$ degenerate to scalars (which we denote by $\check \omega$ and $\check v$, respectively). The expression of (\ref{abf.exp}) is reduced to 
\begin{equation} \label{simple.abf.exp}
   \ABF(\omega, \kappa) = \sqrt{\frac{\check v}{\check v + \check \omega}} \exp \left(\frac{1}{2} \frac{\check \omega}{\check v + \check \omega} \frac{\check \beta^2}{\check v} \right),
\end{equation}
which has the same functional form as the ABF discussed in \cite{Wakefield2009}. 

Although all suitable $\kappa$ values yield the same asymptotic error bound, they have practical implications on the approximation accuracy for finite samples. Our numerical experiments  (Appendix B) indicate that with sample size around hundreds, the $\ABF$s with $\kappa = 0 $ and $\kappa = 1$ both become quite accurate.

\subsection{Connection with frequentist test statistics}

\subsubsection{Connection with fixed effect test statistics}
Consider a specific class of prior, $\Wv = c \Vv$, for which the $\ABF$ can be simplified to
\[ \ABF(\Wv = c \Vv, \kappa) = \left(\sqrt{\frac{1}{c + 1}}\right)^p \exp \left(\frac{1}{2}\frac{c}{c+1} \check \bv' \check \Vv^{-1} \check \bv \right).\]
Consequently, the $\ABF$ becomes a monotonic transformation of the quadratic form $\check \bv' \check \Vv^{-1} \check \bv$. 
We note that, in the following two special cases, the quadratic form corresponds to some popular frequentist statistics to test $\bv$ as a {\em fixed} effect. Particularly, when $\kappa = 1$, the quadratic form becomes the (multivariate) Wald statistic $\hat \bv' \hat \Vv^{-1} \hat \bv$; as $\kappa$ is set to 0, it coincides with the Rao's score statistic (Appendix C.1). 

The monotonic correspondence between the $\ABF$ and these two popular frequentist test statistics indicates that, under the prior specified, the $\ABF$ ranks candidate models (or SNP associations in single-SNP analysis) exactly the same way as both the Wald statistic (for $\kappa=1$) and the score statistic (for $\kappa =0$). 
Furthermore, applying the strategy of Bayes/non-Bayes compromise \citep{Good1992, Servin2007} by treating the $\ABF$ as a regular test statistic, it becomes obvious that the $\ABF$ possesses a $p$-value {\em identical} to that of the corresponding Wald or score statistic, depending on the $\kappa$ values. 
\cite{Wakefield2009} first named the prior specification of the kind $\Wv = c \Vv$ as the {\em implicit p-value prior}.
In the special case of single SNP association testing and assuming Hardy-Weinberg equilibrium, it follows that $\Vv \propto \frac{1}{nf(1-f)}$ (where $f$ represents the allele frequency of a target SNP).
As a consequence, the implicit $p$-value prior essentially assumes a larger {\it a priori} effect for SNPs that are less informative (either due to a smaller sample size or minor allele frequency). Although, from the Bayesian point of view, there seems to be a lack of proper justification for such prior assumptions \citep{Wakefield2009, Wen2014b}, we  often note that the overall effect of the implicit $p$-value prior on the final inference may be negligible in practice, especially when the sample size is large (see section 5.1.1 for illustration). 

\subsubsection{Connection with the variance component score statistic}

In SNP set analysis, it has become common practice to construct a variance component score test for the genetic effect $\bv$  \citep{Wu2011, Lee2012, Schifano2012}. That is, for a set of $p$ SNPs, the genetic effects are assumed to be {\em random} and follow the distribution $\bv \sim {\rm N}( \bf{0},  \gamma \Mv)$,
where the matrix $\Mv$ is pre-defined. 
To test $H_0: \gamma = 0 $ vs. $H_1: \gamma \ne 0$, the score statistic is given by $T_{\rm score} =  \tilde \tau^2 (\yv-\Xv \tilde \av)'\tilde \Sigv^{-1} \Gv \Mv \Gv' \tilde \Sigv^{-1} (\yv-\Xv \tilde \av),$ where $\tilde \Sigv = \Iv + \tilde \lambda \Kv$.
In the special case that the random effect $\uv$ is ignored (i.e., $\lambda =0, \tilde \Sigv = \Iv$),  $T_{\rm score}$ is reduced to the form of the original SKAT statistic \citep{Wu2011}. 
By re-parameterizing $\Wv = \gamma \Mv$, we show that $\ABF(\kappa =0)$ can be represented as a function of $T_{\rm score}$ (Appendix C.2). 
In particular, as $\gamma \to 0$,  it follows that $$  \ABF(\Wv =  \gamma \Mv, \kappa=0) \approx \exp \left( \frac{\gamma}{2}\, T_{\rm score} \right).$$ 
That is, $\ABF(\kappa=0)$ becomes monotonic to the variance component score statistic. 
Interestingly, the condition $\gamma \to 0$ represents a {\em local alternative} scenario (i.e., $\bv$ only slight deviates from $\bf 0$), for which score tests are known to be most powerful.

\section{Genetic Association Analysis with Bayes Factors}

\subsection{Bayesian Hypothesis Testing}

Bayes factors present two major advantages in the hypothesis testing of genetic association signals: namely, the convenience of Bayesian model averaging and the flexibility of utilizing useful prior information.  Before we delve into the details of the advantages of Bayesian models in hypothesis testing, it is worth noting that the practical usage of Bayesian model comparison in hypothesis testing is limited, mostly due to the difficulty involved in determining significance thresholds based on Bayes factors. Traditionally, this issue has been addressed by treating a Bayes factor as a regular test statistic and deriving its $p$-value accordingly \citep{Good1992, Servin2007}. 
Because the null distribution of a Bayes factor is generally non-trivial, most practical implementations rely on permutation procedures. 
Recently, \cite{Wen2013} proposed a robust Bayesian false discovery rate (FDR) control procedure that directly uses the Bayes factors as inputs. This procedure ensures FDR control, even under the mis-specification of alternative models, a property resembling the behavior of $p$-value based procedures under similar circumstances. 
Most importantly, this procedure is highly computationally efficient and generally does not require extensive permutations. 

\subsubsection{Model Averaging}

In hypothesis testing, there often exist multiple alternative scenarios, and a single parametric model (or its corresponding test statistic) can hardly accommodate all cases. 
For example, in SNP set testing of rare-variant genetic associations, there exist two primary types of competing approaches that target different alternative scenarios. 
The first type, represented by the burden tests \citep{Madsen2009}, collapses the genetic variants in a region to form a single characteristic genetic unit, with respect to which the association test is then performed. 
This approach is ideal for a particular alternative scenario in which most of the variants considered are either {\em consistently} deleterious or {\em consistently} protective. 
The second type of the approach, represented by the C-alpha \citep{Neale2011} and SKAT tests, targets a complementary scenario in which the variants included in the SNP set can have bi-directional effects on the phenotype of interest. 
In practice, because the true alternative model is never known {\it a priori}, it remains a challenge to reconcile/combine the results from the two distinct approaches into the frequentist testing paradigm. 
Bayesian model averaging provides a principled way to naturally address this issue.
Suppose that there are $k$ possible alternative models in consideration, and for each model $i$, a Bayes factor $\BF_i$ can be computed and a prior probability/weight, $\pi_i$ is assigned. An overall Bayes factor then can be computed by $\overline \BF = \sum_{i=1}^k \pi_i \BF_i$, which summarizes the overall evidence from the data compared to the null model while accounting for the uncertainty of the true alternative scenario. 

In the context of SNP set analysis, \cite{Lee2012} showed that the alternative scenarios considered in the burden and SKAT tests can both be represented in the LMM framework with different specification of random effect $\Wv$ matrix.
In brief, let the column vector $\wv = (w_1,...,w_p)$ denote the marginal prior effect sizes for $p$  SNPs in a set. 
The burden test assumes $\Wv = \Wv_b = (\sqrt{\wv}) (\sqrt{\wv})'$, whereas the SKAT model assumes $ \Wv = \Wv_s = {\rm diag}(\wv)$. 
Given these results and within the framework of BLMM, we can straightforwardly average the evidence over the the two competing alternative models by computing an overall Bayes factor,  
$\overline \BF (\pi) = \pi \cdot \BF(\Wv_b) + (1- \pi) \cdot \BF(\Wv_s),$
where the probability $\pi$ denotes the relative prior frequency of the burden model. 
Without prior preference over the two alternatives, a natural ``objective" choice is to set  $\pi = 0.5$.

\cite{Lee2012} provided an alternative interpretation by connecting the two models. They considered a class of $\Wv$ matrices indexed by a non-negative correlation coefficient $\rho$: namely,
\begin{equation}
   \Wv_\rho =  {\rm diag}(\sqrt{\wv})\left[ (1-\rho) \Iv + \rho \lv \lv' \right] {\rm diag}(\sqrt{\wv}) = (1-\rho)\, \Wv_s + \rho \, \Wv_b, 
\end{equation}
which we will refer to as the SKAT-O prior. It should be noted that the prior distribution for $\bv$ assumed by Bayesian model averaging is essentially a normal mixture, which itself is not necessarily normal and hence differs from the SKAT-O prior. Nevertheless, the SKAT-O prior can be viewed as a normal approximation of this mixture distribution (to the first two moments). 

\subsubsection{Informative Prior}
The explicit specification of the prior distribution on $\bv$ for alternative models  is seemingly a distinct feature of Bayesian hypothesis testing. However, as we have shown, even the most commonly applied frequentist test statistics can be viewed as resulting from some implicit Bayesian priors. Therefore, it is only natural to regard the prior specification of $\bv$ as an integrative component in alternative modeling. This fact should encourage practitioners to explicitly formulate appropriate informative priors in Bayesian hypothesis testing: if the prior does capture some essence of reality, it improves the overall statistical power; even if the prior is mis-specified, testing with Bayes factors using the procedures, such as either the Bayes/Non-Bayes compromise or the robust Bayesian FDR control, only results in a reduction in power but no inflation of type I error.

For SNP set analysis, it has become common practice to pre-define some  ``weight" for each individual participating SNP in both the burden and SKAT types of approaches (i.e., the aforementioned $\wv$ vector). Most commonly, these priors are set up to prioritize genetic variants with low allele frequencies. 
When performing genetic association analysis, it is now becoming increasingly popular to incorporate genomic annotation and/or pathway information. In all of these examples, BLMM provides a convenient way to formally integrate the prior information into the hypothesis testing.       

Finally, we note that there exist practical settings, especially in the studies of genome-wide scale, in which the information of the desired priors can be sufficiently ``learned" from data facilitated by the Bayes factors. 
Take, for example, the problem of SNP set analysis with two competing alternatives, and consider inferring the weights of the burden and the SKAT models ($\pi$) from the data. 
Hypothetically, if (i) many SNP sets are investigated (in a single or multiple studies) and (ii) a sufficient amount of modest to strong signals are presented in the data,
it should be intuitive that $\pi$ can be accurately estimated by pooling the information across all SNP sets. 
More specifically, for each SNP set, we can augment a latent indicator to represent the true generative model of the observed data. Subsequently, a straightforward EM algorithm (where the complete data likelihood can be evaluated via Bayes factors) can be used to estimate $\pi$.

\subsection{Bayesian Variable Selection in the BLMM}\label{BVS.sec}

Beyond hypothesis testing, many practical problems in genetic association studies can be tackled using model comparison/selection techniques via Bayes factors. 
Here, we consider the problem of multi-locus fine-mapping analysis. In practice, the fine-mapping analysis usually focuses on relatively small genomic regions  flagged by SNP association signals, with the aim of identifying multiple potential signals and narrowing down the candidate causal variants within a region while accounting for LD. 

Consider a region of $p$ candidate variants whose genetic effects are jointly modeled by the $p$-vector $\bv$. Ultimately, we are interested in making an inference on the binary vector $\gav:= \left(\lv(\beta_1 \ne 0), \dots, \lv(\beta_p \ne 0) \right)$.
Under the BLMM, we assume the following spike-and-slab prior for variable selection, namely,
\begin{equation}
   \bv \mid \gav \sim  {\rm N}( {\bf 0}, \Wv) ~{\rm with}~ \Wv = \phi^2 \, {\rm diag}\left(\gav\right),~ {\rm and }~  \Pr\left(\gav\right) = \prod_{i=1}^p p_1^{\xi_i}(1-p_1)^{(1-\xi_i)},
\end{equation}
where the parameter $p_1$ denotes the prior inclusion probability  of a SNP and the parameter $\phi^2$ represents the prior genetic effect size of each SNP. The posterior distribution of $\gav$ can be computed by
\begin{equation}
   \Pr(\gav \mid \phi^2, \yv, \Xv, \Zv, \Gv) \propto \Pr(\gav) \cdot P(\yv \mid \gav, \phi^2, \Xv, \Zv, \Gv) \propto \Pr(\gav) \cdot \BF (\Wv),                                       
\end{equation}
where the Bayes factor can be further approximated by $\ABF(\Wv,\kappa)$. 
It is then conceptually straightforward to design an MCMC algorithm to perform Bayesian variable selection.
We note that, in the case of setting $\kappa = 0$, there are substantial computational savings in the proposed MCMC computation. We give the detailed description and explanation of the MCMC algorithm in Appendix D.

\section{Numerical Illustration}

\subsection{Application of BLMM to an \ath Data Set} 

In this example, we apply the BLMM to study the genetic associations between the genotypes of an inbred \ath line and the quantitative phenotype of sodium concentration in the leaves using the data described in \citep{Baxter2010}.
The data set consists of 336 inbred individuals, and each individual is genotyped at 214K SNP positions genome-wide. The data set was previously analyzed by \cite{Segura2012} under the LMM setting.
We conduct an additional quantile normalization step for the original phenotype measurements to prevent the influence of potential outliers.    

\subsubsection{Single SNP Association Analysis}

We first perform single SNP association tests using the approximate Bayes factors of the BLMM and compare the results with the analyses based on $p$-values. 
To specify the alternative models in the BLMM, we consider a natural exchangeable prior on the standardized effect scale, i.e.,  $\sqrt{\tau} \,\beta \sim {\rm N}(0, \phi^2)$. 
Unlike the implicit $p$-value prior, this prior does not assume a relationship between the genetic effect size and the features of a target SNP. 
Furthermore, instead of fixing a single $\phi$ value, we assume $\phi$ is uniformly drawn from the set $L:=\{\phi: 0.1, 0.2, 0.4, 0.8,1.6\}$, where the various levels of $\phi$ values cover a range of small, modest to large potential effect sizes. 
The use of multiple $\phi$ values forms a mixture normal prior, which is helpful for describing a longer-tailed distribution of effect sizes \citep{Servin2007, Wen2014a}. The range of the $\phi$ values is selected following the suggestion of \cite{Balding2009}.
We use the software package GEMMA \citep{Zhou2012} to estimate the kinship matrix, $\Kv$, for the random effect, and obtain the MLEs, $\hat \beta (\hat \lambda), \hat \beta(\tilde \lambda)$, along with their standard errors for all the SNPs. Applying the equation (\ref{simple.abf.exp}), we then compute the approximate Bayes factors at $\kappa=1$ and $\kappa=0$ for each $\phi_i$ value. Finally, we compute an overall Bayes factor by averaging over all the prior effect size models, i.e., $\BF = \frac{1}{||L||} \sum_i \BF(\phi_i)$.

We first investigate the ranking of the association signals by the ABFs under the natural Bayesian prior and the $p$-values based on the score and Wald test statistics. To this end, we compute the Spearman's rank correlation coefficient ($\rho$) of the $\log_{10}(\ABF)$ and $-\log_{10}(\mbox{$p$-value})$.
The overall rank correlation (from all 214K association tests) between $-\log_{10}(\mbox{$p$-value})$ based on the score statistic and $\log_{10}[\ABF(\kappa=0)]$ is 0.817. However, we note that the majority of the discordance in ranking comes from the unlikely association signals (see Figure \ref{abfvsp.fig}), which are generally not of interest. 
Focusing on the subset of 10,913 SNPs with $p$-value $< 0.05$, the rank correlation becomes nearly perfect ($\rho = 0.995$). 
Similarly, the $-\log_{10}(\mbox{$p$-value})$ based on the Wald statistic has an overall rank correlation of 0.821 with $\log_{10}[\ABF(\kappa=1)]$, and for the subset of 11,379 SNPs with corresponding $p$-value $< 0.05$, $\rho = 0.996$. The direct comparison between the approximate Bayes factors and corresponding $p$-values is shown in Figure \ref{abfvsp.fig}.

\begin{figure}[!ht]
\begin{center}
\includegraphics[totalheight=0.42\textheight]{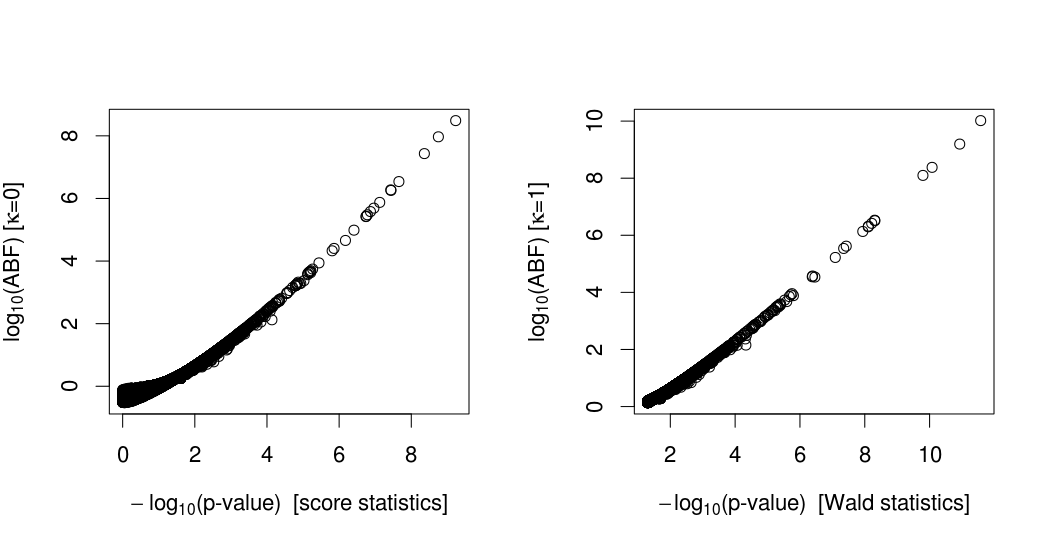}
\caption{\label{abfvsp.fig} Direct comparison of the $\ABF$s and $p$-values on  the log scale. The plot shows that the rankings of the association signals based on the Bayes factor and the $p$-value are largely in agreement, especially for SNPs showing modest to strong signs of association. }
 \end{center}
\end{figure}

As an illustration, we further apply the Bayesian and the frequentist FDR control procedures for the Bayes factors and $p$-values to determine the significance cut-offs, ignoring correlations among the tests. 
Ultimately, both the Benjamini-Hochberg and the Storey procedures using the score statistic $p$-values select 17 significant SNPs (denoted by set $S_p$).
In comparison, the standard Bonferroni procedure selects 12 SNP (denoted by set $S_b$). 
The Bayesian FDR control procedure (i.e., the EBF procedure, described in \cite{Wen2013}) based on $\ABF(\kappa =0)$ selects 14 significant SNPs (denoted by set $S_{BF}$).  
Importantly, we note that $S_b \subset S_{BF} \subset S_p$. The results from the $\ABF(\kappa=1)$ and Wald statistic $p$-values are nearly identical.  

Based on this result, we conclude that, under this particular GWAS setting with a very modest sample size, there is no obvious practical difference in applying the Bayes factors and the $p$-values in single SNP hypothesis testing. We view this result as a numerical validation of our theoretical results discussed in section 3.1.

\subsubsection{Fine-Mapping Analysis}

Following \cite{Segura2012}, we further perform a multi-locus fine-mapping analysis of a 200kb genomic region centered around the top single SNP association signal at chr4:6392280, where 508 SNPs are included. 
Using the MCMC algorithm described in section \ref{BVS.sec}, we assign the prior inclusion probability  $p_1 = 1/508$ for each candidate SNP, which conservatively sets the prior expected number of signals in the region to 1. Conditional on a SNP having a non-zero effect (i.e., $\lv(\beta_i) \ne 0$), we use the same normal mixture prior for the effect size $\beta_i$ described in the single SNP association analysis.
We obtain the posterior samples from 300,000 MCMC repeats after 150,000 burn-in steps, and the convergence of the MCMC algorithm is diagnosed using the procedure described in \cite{Brooks2003}. 

The analysis based on the posterior samples clearly indicates that there are multiple independent association signals residing in this relatively small genomic region.   There is zero probability mass on those posterior models containing fewer than 3 SNPs; the probabilities for the posterior models having 3, 4, 5 and 6 independent signals are 0.175, 0.452, 0.350 and 0.023, respectively. 
Inspecting individual SNPs, we summarize the top five associated SNPs according to their posterior inclusion probabilities in Table \ref{post.incp.tab}. 
The correlations among the top 5 SNPs are very modest. 
Thus far, our result has been largely consistent with what is reported in \cite{Segura2012}, in which a stepwise variable selection scheme with a BIC-like model selection criteria is employed. 
Nevertheless, we notice  a great deal of uncertainty within the individual models from our analysis. 
The details of the top 10 models ranked by their posterior probabilities are shown in Table \ref{post.model.tab}.
The {\it maximum a posterior} (MAP) model only has a probability of 0.05, and all of the top models have similar complexities and very comparable likelihoods.
In addition, we find that 61\% of the posterior models contain both of the top two SNPs, and 32\% of the posterior models contain a combination of the top three SNPs.
One may naturally suspect that the uncertainty in relative large models (i.e., with more SNPs included) is partially due to the stringent $p_1$ prior. To this end, we modify the prior distribution to $\log_{10} p_1\sim {\rm Uniform}[-2.71, -1.40]$ (the two end points correspond to $p_1$ equaling $1/508$ and $20/508$, respectively), but the results do not qualitatively change. 
Biologically, it might be the case that the true causal variants are not directly genotyped and the observed signals are only partially correlated with them. It is then worth following up with dense genotyping experiments or genotype imputations. 
Statistically, it seems evident that, in this particular case, reporting a single ``best" model from the variable selection procedure yields an over-simplified picture and can be misleading for the follow-up analysis.

\begin{table}[!ht]
\begin{center}
{\fontsize{10}{12}\selectfont 
\begin{tabular} { c    c     c }
\hline
SNP & Posterior Inclusion Prob. & Marginal $\log_{10}(\ABF)$ \\
\hline
\hline
chr4:6414956 & 0.795 & 4.98 \\
chr4:6392280 & 0.741 & 7.96 \\
chr4:6420777 & 0.528 & 6.03 \\
chr4:6455695 & 0.451 & 5.30 \\
chr4:6391204 & 0.405 & 7.92 \\
\hline
\end{tabular}}
\caption{\label{post.incp.tab} Top 5 associated SNPs according to their marginal inclusion probabilities in the Bayesian fine-mapping analysis. The last column shows the values of $\log_{10} \ABF(\kappa=0)$ from the single SNP association testing. Only SNP chr4:6392280 and SNP chr4:6391204 show a very modest LD, whereas all of the other pairs of SNPs are very weak in LD. }

\end{center}
\end{table}

\begin{table}
\begin{center}{\fontsize{10}{12}\selectfont 
\begin{tabular} { l    c     c }
\hline

Model  & Posterior  Prob. &  $\log_{10}(\ABF)$ \\
\hline
\hline
 chr4:6392280 + chr4:6394774 + chr4:6414956 + chr4:6421034     & 0.052   &   19.55 \\
 chr4:6392280 + chr4:6414956 + chr4:6420777 + chr4:6455695     & 0.039   &   18.89 \\
 chr4:6391204 + chr4:6392280 + chr4:6414956 + chr4:6420777     &  0.032  &   18.47 \\    
 chr4:6380552 + chr4:6391204 + chr4:6414956 + chr4:6455695     &  0.028  &   18.67 \\     
 chr4:6391204 + chr4:6414956 + chr4:6420777 + chr4:6455695     & 0.026   &   18.72 \\     
 chr4:6392280 + chr4:6414956 + chr4:6418442 + chr4:6420777     &  0.024  &   18.68 \\    
 chr4:6391286 + chr4:6392280 + chr4:6414956 + chr4:6420777     &  0.022  &   18.35 \\     
 chr4:6392280 + chr4:6414956 + chr4:6418442                                &  0.018  &   16.75 \\   
 chr4:6380552 + chr4:6391204 + chr4:6392280 + chr4:6420777     &   0.017  &   18.19 \\    
 chr4:6380552 + chr4.6392280 + chr4:6394774 + chr4:6414956 + chr4:6421034 & 0.016 & 21.25 \\
\hline
\end{tabular}}
\caption{\label{post.model.tab} Top 10 posterior models in the Bayesian fine-mapping analysis. The models are ranked according to their posterior probabilities (second column). The last column shows the values of $\log_{10} \ABF(\kappa=0)$ of the corresponding models. Our prior specification encourages sparse models: complicated models with more predictors are penalized more severely by the prior inclusion probability. The most important feature of these results is that there is not a unique simple model that is clearly better than the others. }
\end{center}
\end{table}

\subsection{Simulation Study of SNP Set Analysis}
In this section, we perform simulation studies to illustrate the effectiveness of the proposed Bayesian model comparison approach in SNP set analysis. In each simulated data set, we generate 5,000 phenotype-SNP set pairs that mimics the data structure from genome-wide investigation of expression quantitative trait loci (eQTLs). We randomly select 3,500 SNP sets and simulate their phenotypes from a null model. For the remaining SNP sets, we use two types of alternative models described in \cite{Lee2012} to generate their phenotypes: one model assumes consistent directional effects of rare variants, whereas the other allows inconsistent directional effects. We use $\pi$ to denote the relative frequency of the sign-consistent models in all the alternative models, and vary this parameter in different simulation sets. We give a detailed account of the simulation schemes in Appendix E.1.

We analyzed the simulated data sets using the proposed Bayesian model comparison approach and the SKAT-O method implemented in the R package SKAT (version 0.95) to examine their controls of FDR and powers. 
For both approaches, we again follow the previous work \citep{Wu2011,Lee2012} and assign the marginal weight for each SNP as a function of their allele frequencies.
In Bayesian analysis, we apply two strategies in choosing the prior weights for Bayes factor computation. 
The first strategy assumes an ``objective" uniform prior setting $\pi = 0.5$, and the second strategy {\em estimates} $\pi$ and the distribution of genetic effect sizes  from the data by pooling information across all phenotype-SNP set pairs using a hierarchical model. The details of the analysis procedure are provided in Appendix E.2.

We summarize the simulation results in Table \ref{set.sim.tbl}. The false discovery rates in all the methods are well controlled. The performance of the Bayesian procedure with the default uniform prior weights is very similar to that of the SKAT-O, and the Bayesian procedure based on informative priors achieves the best power in  all settings defined by different true $\pi$ values. 
These results are well expected because the Bayesian method with estimated weights has the unique advantage of effectively borrowing information across genes through the use of Bayes factors and hierarchical modeling. 
In addition, we want to emphasize that all of the Bayesian models assumed in the analysis are indeed very ``wrong" comparing to the true data-generating model; nevertheless, the robust Bayesian FDR control procedure using Bayes factors ensures the targeted FDR level.

\begin{table}
\begin{center}
{\fontsize{10}{12}\selectfont 
\begin{tabular} { c  c   c c c  c  c c c }
\hline

   ~ & ~& \multicolumn{3}{c}{FDR} &~& \multicolumn{3}{c}{Power}\\
\cline{3-5} \cline{7-9}
 Setting ($\pi$) & ~& SKAT-O & Bayesian-D & Bayesian-E & ~ & SKAT-O &  Bayesian-D & Bayesian-E \\
\hline
\hline
 0.20 & ~ &  0.024 & 0.027 & 0.023 & ~ & 0.768 & 0.741 & 0.821 \\
 0.40 & ~ &  0.046 & 0.030 & 0.028 & ~ & 0.791 & 0.773 & 0.828 \\
 0.50 & ~ &  0.051 & 0.045 & 0.041 & ~ & 0.836 & 0.825 & 0.869 \\
 0.60 & ~ &  0.049 & 0.046 & 0.045 & ~ & 0.909 & 0.908 & 0.919 \\
 0.80 & ~ &  0.050 & 0.049 & 0.048 & ~ & 0.933 & 0.943 & 0.948 \\
\hline

\end{tabular}}
\caption{\label{set.sim.tbl} Realized false discovery rate and power in simulation studies of SNP set analysis. The first column (setting) indicates the percentage of the SNP sets with sign-consistent effects among all of the non-null SNP sets in the simulated data. For the SKAT-O procedure, the resulting $p$-values are further processed by the Storey procedure for FDR controls. ``Bayesian-D" indicates the Bayesian testing procedure with the default uniform weights. "Bayesian-E" indicates the Bayesian procedure that estimates $\pi$ from the data. The FDR control for the Bayes factors is performed using the EBF procedure described in \cite{Wen2013}.   }

\end{center}
\end{table}

Going beyond SNP set testing targeting rare variant associations, in Appendix F, we further demonstrate that our Bayesian model averaging framework can be conveniently extended to integrate models for detecting  {\em common  variant associations} into SNP set testing. We envision that this approach will have a profound impact in  studies of expression trait quantitative loci at genome-wide scale. 
 
\section{Discussion}

In this paper, we have presented a unified Bayesian framework to perform model comparisons in the contexts of a linear mixed model and SNP set analysis. 
Although our statistical results are presented exclusively for the quantitative response variables, it is possible to extend them to the generalized linear mixed models (GLMM) context to incorporate binary outcomes and count data using a quadratic approximation of the corresponding log-likelihood functions. 

Primarily based on the results of the approximate Bayes factors, we have demonstrated an efficient Bayesian sparse variable selection algorithm to perform multi-locus association analysis using the BLMM. 
Recently, \cite{Zhou2013} also proposed an elegant Bayesian solution for multiple SNP association analysis under the LMM model on the genome-wide scale. 
It should be noted that their method also has a primary focus on estimating the  heritability, whereas our method is designed for fine-mapping analysis. 
In addition, by treating SNP sets as selection units, our approach can be straightforwardly extended to the identification of multiple associated genes/SNP sets, which may be attractive for biological pathway analysis. 
Previous studies \citep{Guan2011, Wen2014a} have shown that Bayesian methods generally hold advantages over penalized regression approaches in variable selection problems with correlated covariates (e.g., SNPs in LD) and/or non-i.i.d. residual error structures. 
More importantly, as we have demonstrated, there can be great uncertainty regarding any single ``best fitting" model. As a practical consequence, reporting a single ``best" model but ignoring appropriate uncertainty assessments could hinder follow-up scientific investigation. 

Finally, we want to note that Bayesian model comparison approaches have been successfully assessed in other areas of genetic association studies, e.g., meta-analysis \citep{Wen2014b}, association mapping of multiple-traits and detecting gene-environment interactions \citep{Flutre2013, Wen2014b}. Our results can be conveniently integrated into those existing tools, and their usages can be naturally extended to incorporate LMM and SNP set analysis.

\section{Supplementary Material}

The software, scripts used to generate simulated data can be found at \url{http://github.com/xqwen/BLMM}. Detailed derivations, proofs and descriptions of relevant algorithms and simulation details are included in the supplementary file.

\section*{Acknowledgments}

We thank Seunggeun Lee and Xiang Zhou for helpful discussions. This work is supported by NIH grants R01-MH101825 and R01-HG007022. 

\newpage

\appendix

\section{Bayes Factor Derivation}

In this section, we show the detailed derivation of the approximate Bayes factors under the BLMM, which also serves as a proof for Proposition 1 in the main text. 

\subsection{Exact Bayes factor with known $\lambda$ and $\tau$}
We first consider the case where the variance parameters $\tau$ and $\lambda$ are known, instead of being assigned priors. In this case, we show that the exact Bayes factor under the BLMM can be analytically computed. We summarize this result in the following lemma:

\begin{lemma}
Under the BLMM, if the variance parameters $\tau$ and $\lambda$ are known, the Bayes factor can be analytically computed by
\begin{equation}\label{exact.bf}
   \BF(\Wv) = | \Iv + \hat \Vv^{-1} \Wv |^{-\frac{1}{2}} \cdot \exp \left( \frac{1}{2} \hat \bv'  \hat \Vv^{-1} \left[ \Wv ( \Iv +  \hat \Vv^{-1} \Wv)^{-1} \right] \hat \Vv^{-1} \hat \bv \right).
\end{equation}
\end{lemma}

\begin{proof}
The linear mixed model can be equivalently represented by
\begin{equation}
   \label{full.model}
   \begin{aligned}
  \yv &= \Xv \av + \Gv \bv  + \epiv, \\
  \epiv &\sim {\rm N}\left(0, \tau^{-1} \Sigv \right),
  \end{aligned}
\end{equation}
where $\Sigv = \Iv + \lambda \Kv $. With the variance parameters and $\Sigv$ known, we perform the following transformations to the observed data:
\begin{equation}
  \begin{aligned}
   & \utilde{\yv} = \Sigv^{-\frac{1}{2}}\yv \\
   & \utilde{\Xv} = \Sigv^{-\frac{1}{2}}\Xv \\
   & \utilde{\Gv} = \Sigv^{-\frac{1}{2}}\Gv \\
  \end{aligned}
\end{equation}  
This results in a linear model
\begin{equation}
   \label{trans.lmm}
    \begin{aligned}
  \utilde{\yv} &= \utilde{\Xv} \av + \utilde{\Gv} \bv  + \utilde{\epiv}, \\
  \utilde{\epiv} &\sim {\rm N}\left(0, \tau^{-1} \Iv \right),
  \end{aligned}
\end{equation}
where $\utilde{\epiv} = \Sigma^{-\frac{1}{2}} \epiv$. Linear model (\ref{trans.lmm}) is a trivial special case of the complex linear model systems considered by \cite{Wen2014a}. Consequently, it follows from the Lemma 1 of \cite{Wen2014a}, given the prior specifications described in the main text, the Bayes factor can be analytically computed by 
\begin{equation}
  \BF(\Wv; \lambda, \tau) = | \Iv + \hat \Vv^{-1}\Wv|^{-\frac{1}{2}}
                  \cdot \exp \left( \frac{1}{2} \hbvg' \hat \Vv^{-1} \left[\Wv (\Iv + \hat \Vv^{-1}\Wv)^{-1}\right] \hat \Vv^{-1} \hbvg \right).
\end{equation}

\end{proof}

Next, we show the detailed analytic forms of $\hbvg$ and $\hat \Vv$ under the BLMM. 
First, we define 
\begin{equation}
  \Gv_{\xv} = \left(\Iv - \Sigv^{-1/2} \Xv \left(\Xv'\Sigv^{-1}\Xv\right)^{-1} \Xv' \Sigv^{-1/2}\right) \Sigv^{-1/2} \Gv,
\end{equation}
which only depends on $\lambda$ through $\Sigv$.
It follows that
\begin{equation}
  \hbvg(\lambda) = \left(\Gv_{\xv}'\Gv_{\xv}\right)\Gv_{\xv}'\Sigv^{-1/2} \yv,
\end{equation}
and
\begin{equation}
  \hat \Vv(\lambda, \tau) = \tau^{-1} \left(\Gv_{\xv}'\Gv_{\xv}\right)^{-1}.
\end{equation}

\subsection{Approximate Bayes factors for unknown $\lambda$ and $\tau$}

When $\lambda$ and $\tau$ are unknown, to compute the Bayes factor, it is required to evaluate the following marginal likelihood
\begin{equation}
  p(\yv \mid \Wv, \Xv, \Gv, \Zv) = \int p(\yv \mid \Wv, \Xv, \Gv, \Zv, \lambda, \tau) p(\lambda, \tau) d\lambda d\tau,
\end{equation}
and the desired Bayes factor is therefore computed as
\begin{equation}
  \BF(\Wv) = \lim_{\Psiv^{-1} \to 0} \frac{ \int p(\yv \mid \Wv, \Xv, \Gv, \Zv, \lambda, \tau) p(\lambda, \tau) d\lambda d\tau }{ \int p(\yv \mid \Wv = 0, \Xv, \Gv, \Zv, \lambda, \tau) p(\lambda, \tau) d\lambda d\tau}.
\end{equation}
By applying the Bounded convergence theorem (to switch limit and integration), we can carry the analytic computation up to the following point 
  \begin{equation}
   {\rm BF}(\Wv) = \frac{\int K_{H_a} \,d \lambda d \tau }{\int K_{H_0} \,d \lambda \,d \tau},
\end{equation}
where 
\begin{equation}
 \begin{aligned}
  K_{H_a} &= |\Iv + \hat \Vv(\tau, \lambda)^{-1} \Wv(\tau, \lambda)|^{-\frac{1}{2}}  \\ 
                   & \cdot \exp\left(\frac{1}{2}\hbvg(\lambda)'\hat \Vv(\tau, \lambda)^{-1} \Wv(\tau, \lambda)\left[ \Iv +  \hat \Vv(\tau, \lambda)^{-1} \Wv(\tau, \lambda)\right]^{-1} \Vv(\tau, \lambda)^{-1} \hbvg(\lambda)\right)  \\
                   & \cdot \tau^{\frac{n}{2}} |\Sigv(\lambda)|^{-\frac{1}{2}}\, p(\lambda, \tau)\cdot \exp\left(-\frac{\tau}{2}\left[\yv -\Xv \tilde \av (\lambda) \right]' \Sigv(\lambda)^{-1} \left[\yv -\Xv \tilde \av (\lambda) \right] \right),
  \end{aligned}
\end{equation}
and 
\begin{equation}
  K_{H_0} =\tau^{\frac{n}{2}}|\Sigv(\lambda)|^{-\frac{1}{2}}\, p(\lambda, \tau) \cdot \exp\left(-\frac{\tau}{2}\left[\yv -\Xv \tilde \av (\lambda) \right]' \Sigv(\lambda)^{-1} \left[\yv -\Xv \tilde \av (\lambda) \right] \right),
\end{equation}
where 
\begin{equation} \label{av.exp}
  \tilde \av(\lambda) = (\Xv'\Sigv^{-1}\Xv)^{-1} \Xv'\Sigv^{-1} \yv.
\end{equation}

We propose to approximate the double integrals of both $K_{H_a}$ and $K_{H_0}$ by Laplace's method. In general, Laplace's method approximates a multiple integral with respect to a $p$-vector $\zv$ in the following fashion,
\begin{equation} \label{lap.appx}
  \int_D h(\zv)\exp\left[g(\zv)\right] d\zv \approx (2 \pi)^{p/2} |\Hv_{\hat \zv}| h(\hat \zv) \exp\left[g(\hat \zv)\right],
\end{equation} 
where 
\begin{equation*}
 \hat \zv = \arg\max_{\zv} g(\zv), 
\end{equation*}
and $|\Hv_{\hat \zv}|$ is the absolute value of the determinant of the Hessian matrix of the function $g$ evaluated at $\hat \zv$. There may be multiple choices to factor an integrand into functions $h$ and $g$, the technical requirements for a valid asymptotic approximation are 
\begin{enumerate}
\item $h$ is smooth and positively valued
\item $g$ is smooth and obtains its unique maximum (w.r.t $\zv$) in the interior of $D$
\item $g$ is linear increasing with respect to the sample size $n$
\end{enumerate}
Different factorization schemes satisfying above requirements usually yield different approximation accuracies for finite sample size, nonetheless, their asymptotic error bounds are the same. For a detailed discussion, see \cite{Butler2007} chapter 2.

We apply a specific factorization for $K_{H_a}$ and $K_{H_0}$ for Laplace's method. First, we note the decomposition of the quadratic form 
\begin{equation}
\begin{aligned}
         &\tau \left[\yv -\Xv \tilde \av (\lambda) \right]' \Sigv(\lambda)^{-1} \left[\yv -\Xv \tilde \av (\lambda) \right]  \\
          &=\hbvg(\lambda) \hat \Vv(\tau, \lambda)^{-1} \hbvg(\lambda)+ \tau \left[\yv -\Xv \hat \av (\lambda) -\Gv \hbvg(\lambda) \right]' \Sigv(\lambda)^{-1}  \left[\yv -\Xv \hat \av (\lambda) -\Gv \hbvg(\lambda) \right],
\end{aligned}  
\end{equation}
where
\begin{equation}
   \left( \begin{array}{c} \hat \av \\  \hbvg \\ \end{array}\right) =   \left[(\Xv~\,\Gv)'\Sigv^{-1}(\Xv~\,\Gv)\right]^{-1} (\Xv~\,\Gv)'\Sigv^{-1} \yv.
\end{equation}
Thus, for an arbitrary weight parameter $\kappa \in [0,1]$, we can write 
\begin{equation}
\begin{aligned}
  &~~~~\tau \left[\yv -\Xv \tilde \av (\lambda) \right]' \Sigv(\lambda)^{-1} \left[\yv -\Xv \tilde \av (\lambda) \right]\\
   &= \kappa \cdot \left( \hbvg(\lambda) \hat \Vv(\tau, \lambda)^{-1} \hbvg(\lambda)+
\tau \left[\yv -\Xv \hat \av (\lambda) -\Gv \hbvg(\lambda) \right]' \Sigv(\lambda)^{-1}  \left[\yv -\Xv \hat \av (\lambda) -\Gv \hbvg(\lambda) \right] \right)\\
   & ~~ + (1-\kappa)\cdot \tau \left[\yv -\Xv \tilde \av (\lambda) \right]' \Sigv(\lambda)^{-1} \left[\yv -\Xv \tilde \av (\lambda) \right]
\end{aligned}
\end{equation}
Using this decomposition, we factor $K_{H_a}$ into $K_{H_a} = h_a(\lambda, \tau) \exp [g_a(\lambda, \tau)]$, where 
\begin{equation}
\begin{aligned}
  &~~h_a(\lambda, \tau)  = |\Iv + \hat \Vv(\tau, \lambda)^{-1} \Wv(\tau, \lambda)|^{-\frac{1}{2}}  \\ 
                   & \cdot \exp\left(\frac{1}{2}\hbvg(\lambda)'\hat \Vv(\tau, \lambda)^{-1} \Wv(\tau, \lambda)\left[ \Iv +  \hat \Vv(\tau, \lambda)^{-1} \Wv(\tau, \lambda)\right]^{-1} \Vv(\tau, \lambda)^{-1} \hbvg(\lambda)\right)\\
            &\cdot \exp\left( -\frac{\kappa}{2}  \hbvg(\lambda) \hat \Vv(\tau, \lambda)^{-1} \hbvg(\lambda) \right) \cdot p(\lambda, \tau)      
 \end{aligned}
\end{equation} 
and
\begin{equation}
\begin{aligned}
 g_a(\lambda, \tau) &= \frac{n}{2} \log(\tau) - \frac{1}{2} \log|\Sigv(\lambda)| - \frac{\tau}{2} (1-\kappa) \left(\left[\yv -\Xv \tilde \av (\lambda) \right]' \Sigv(\lambda)^{-1} \left[\yv -\Xv \tilde \av (\lambda) \right]\right) \\
    & ~ - \frac{\tau}{2} \kappa   \left( \left[\yv -\Xv \hat \av (\lambda) -\Gv \hbvg(\lambda) \right]' \Sigv(\lambda)^{-1}  \left[\yv -\Xv \hat \av (\lambda) -\Gv \hbvg(\lambda) \right] \right).
 \end{aligned}  
\end{equation}
We factorize into $K_{H_0} = h_0(\lambda, \tau) \exp [g_0(\lambda, \tau)]$, where  
\begin{equation}
  h_0(\lambda, \tau) = \exp\left(-\frac{\kappa}{2}  \hbvg(\lambda) \hat \Vv(\tau, \lambda)^{-1} \hbvg(\lambda) \right) \cdot p(\lambda, \tau),      
\end{equation}
and $g_0(\lambda, \tau) = g_a(\lambda, \tau)$.

Note that for any given $\lambda$ value, there is a corresponding $\tau$ value, namely,
\begin{equation}
  \begin{aligned}
   \hat \tau(\lambda; \kappa) &= n \bigg / \bigg\{ (1-\kappa) [\yv -\Xv \tilde \av (\lambda)]' \Sigv(\lambda)^{-1} [\yv -\Xv \tilde \av (\lambda) ] \\
                       & + \kappa [\yv -\Xv \hat \av (\lambda) -\Gv \hbvg(\lambda) ]' \Sigv(\lambda)^{-1} [\yv -\Xv \hat \av (\lambda) -\Gv \hbvg(\lambda)] \bigg\},
  \end{aligned}
\end{equation}
maximizes the $g_a(\lambda, \tau)$ among all possible $\tau$ values. Consequently, maximizing function $g_a(\lambda, \tau)$ is equivalent to maximize $g_a(\lambda, \hat \tau(\lambda))$ with respect to the single parameter $\lambda$. Therefore, we can simplify the target objective function to
\begin{equation}
 \label{obj.func}	
  l(\lambda;\kappa) = \frac{n}{2}\log \hat \tau(\lambda) -\frac{1}{2} \log |\Sigv(\lambda)|.
\end{equation} 
It should be noted that as in the special cases $\kappa=1$ and $\kappa =0$, the objective function (\ref{obj.func}) becomes the score functions of the full and null LMMs, respectively. In general, there is no strong guarantee that the function (\ref{obj.func}) is strictly concave with respect to $\lambda$. Nevertheless, the second derivative of $l(\lambda; \kappa)$ (not shown, see \cite{Zhou2012} for reference) suggests that the objective function is asymptotically concave (i.e., concave for sufficiently large sample size $n$).
There is no analytic solution to optimize (\ref{obj.func}), and the gradient based numerical optimization algorithms, e.g. the Newton-Raphson method, are typically applied in this setting (because the derivatives of the objective functions can be efficiently evaluated, as demonstrated in \cite{Zhou2012}).    
We denote
\begin{equation}
   \check \lambda (\kappa) = \arg\max_\lambda\, l(\lambda;\kappa),
\end{equation}
and 
\begin{equation}
    \check \tau = \hat \tau(\check \lambda).
\end{equation}
Based on (\ref{lap.appx}), Laplace's method yields the following approximation to the Bayes factor
\begin{equation}
  \begin{aligned}
  \BF(\Wv)  & = |\Iv + \hat \Vv(\check \tau, \check \lambda)^{-1} \Wv(\check \tau, \check \lambda)|^{-\frac{1}{2}}  \\ 
                   & \cdot \exp\left(\frac{1}{2}\hbvg(\check \lambda)'\hat \Vv(\check \tau, \check \lambda)^{-1} \Wv( \check \tau, \check \lambda)\left[ \Iv +  \hat \Vv(\check \tau, \check \lambda)^{-1} \Wv(\check \tau, \check \lambda)\right]^{-1} \Vv(\check \tau, \check \lambda)^{-1} \hbvg(\check \lambda)\right)\\
                   & \cdot \left(1+O\left(\frac{1}{n}\right)\right).
  \end{aligned}
\end{equation}
This essentially proves Proposition 1.
    
\section{Numerical Accuracy of Approximate Bayes Factors}

The proposition 1 shows that the approximate Bayes factors under the BLMM have an $O(1/n)$ error bound for $\kappa \in [0,1]$. In this section, we perform numerical experiments to investigate impacts of different values of $\kappa$ and $n$ on the accuracy of the approximations. To this  end, we sub-sample the real genotype and phenotype data from the \ath example to obtain 3,000 SNPs at various sample sizes: $n = 50, 100, 150, 336$. 
For each sub-sampled data set, we compute the $\ABF$s for $\kappa = 0 \mbox{ and } 1$ based on the output from GEMMA using equation (3.7) in the main text. 

In comparison, we compute the Bayes factors by numerical integration. 
For general prior $p(\lambda, \tau)$, the two-dimensional numerical integration is practically implausible. 
(Although Monte Carlo integration is possible, we find the results typically exhibit extremely large variances.) 
To overcome this difficulty, we apply a specific form of prior, $$p(\lambda, \tau) = p(\lambda) p(\tau) \propto \frac{1}{\lambda} \, \frac{1}{\tau},$$
i.e., the priors of $\lambda$ and $\gamma$ are independent, $p(\tau)$ is assumed a limiting gamma distribution and $p(\lambda)$ is assumed a limiting inverse-gamma distribution. 
With this specification, it becomes possible to first  integrate out $\tau$ analytically conditional on $\gamma$ and then perform a one-dimension numerical integration with respect to $\gamma$ using the adaptive Gaussian quadrature algorithm implemented in R. 
It is worth emphasizing that in this exercise, the statistical interpretation of the priors or the resulting Bayes factors is unimportant, we simply attempt to evaluate the numerical differences by different Bayes factor computation methods.  

The comparison results for different sample sizes and $\kappa$ values are summarized in Figure \ref{numerical.comp.fig}. Regarding the results from the numerical integrations as the ``truth", we find that  for relatively small sample sizes,  $\ABF(\kappa=0)$ tends to be slightly conservative while $\ABF(\kappa=1)$ tends to be slightly anti-conservative. However, when the sample size grows $\sim 300$, both approximations become quite accurate. 

\begin{figure}[!ht]
\centering
    \subfloat[sample size $n = 50$]{%
     \includegraphics[width=.40\textwidth]{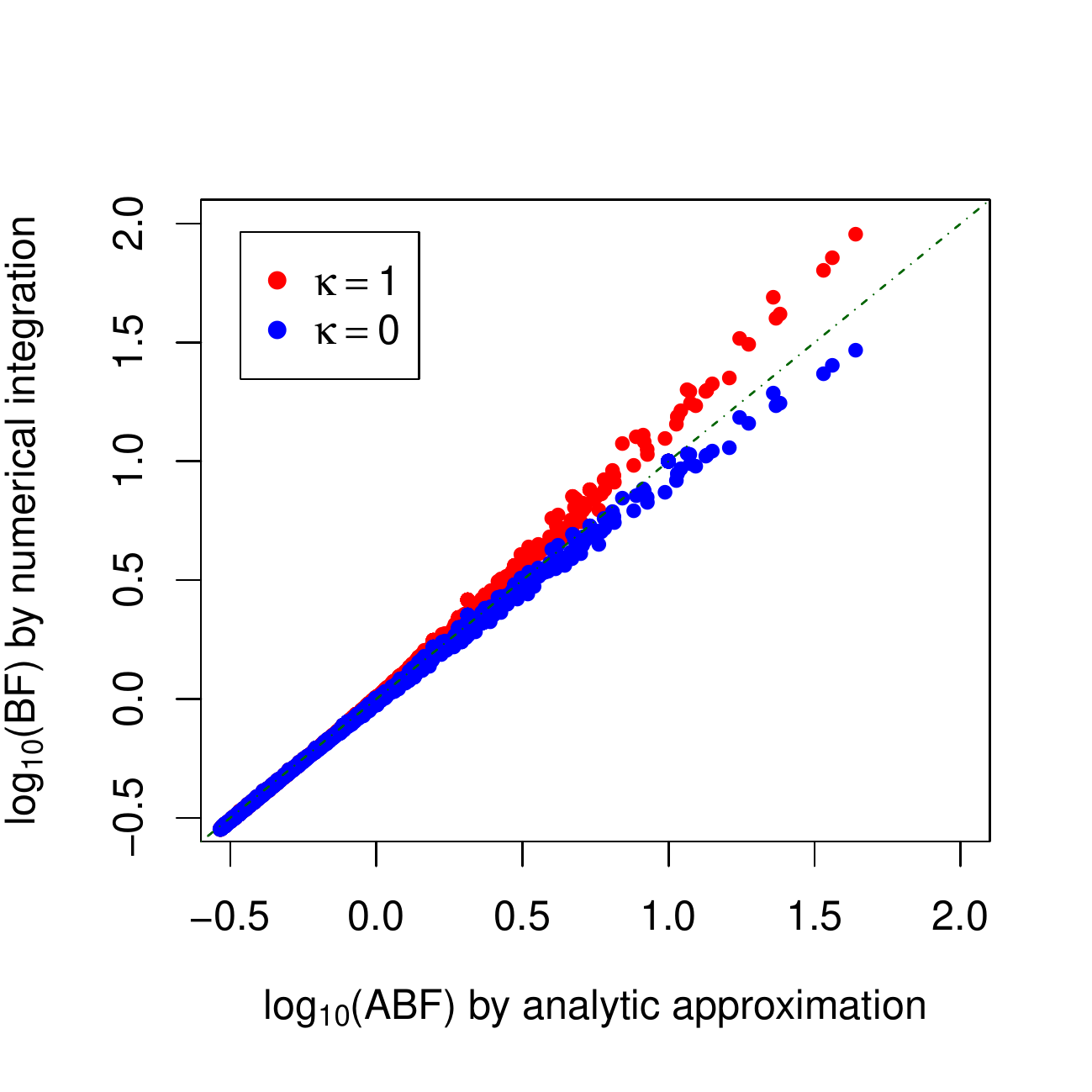}} 
     \subfloat[sample size $n = 100$]{%
    \includegraphics[width=.40\textwidth]{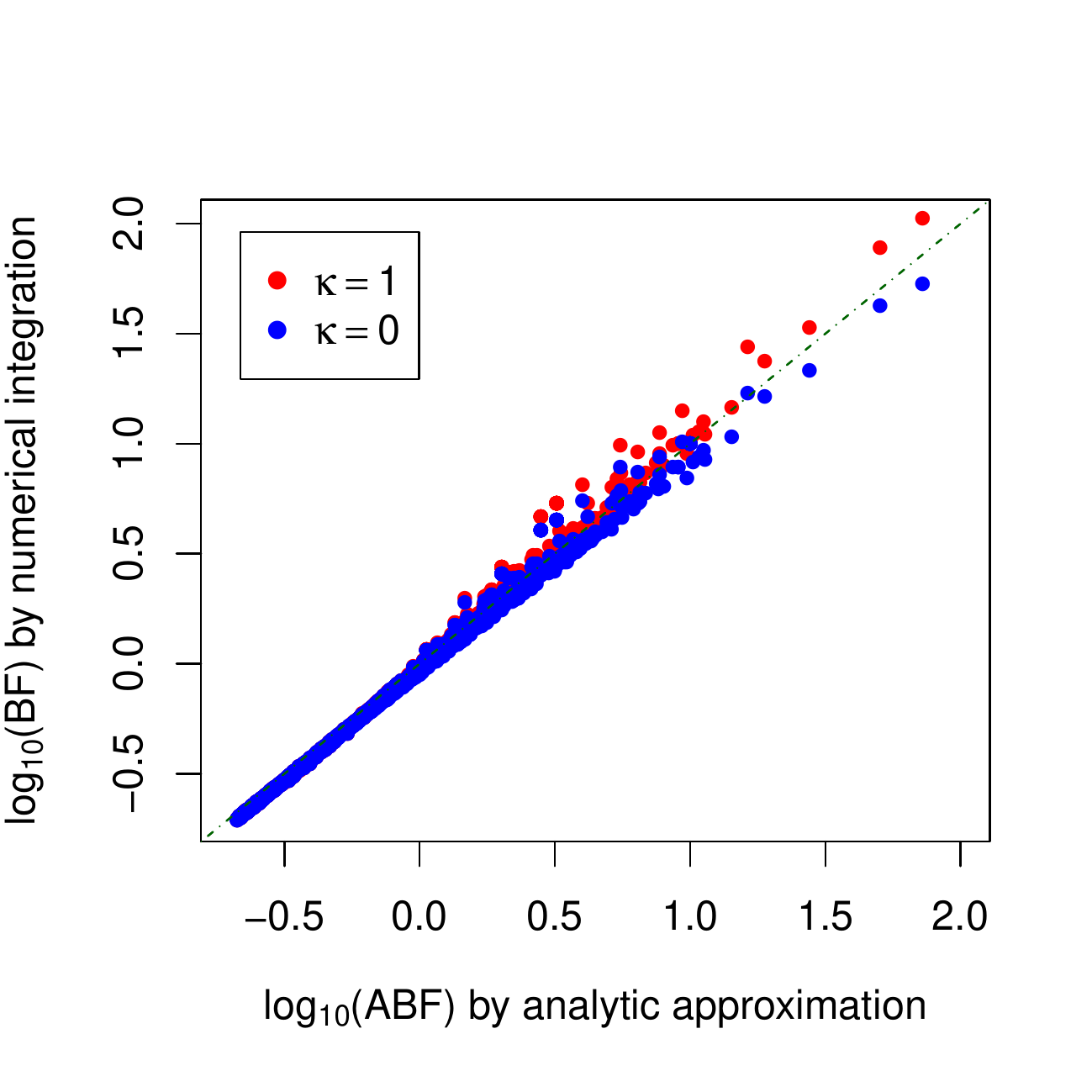}}\\
    \subfloat[sample size $n = 200$]{%
     \includegraphics[width=.40\textwidth]{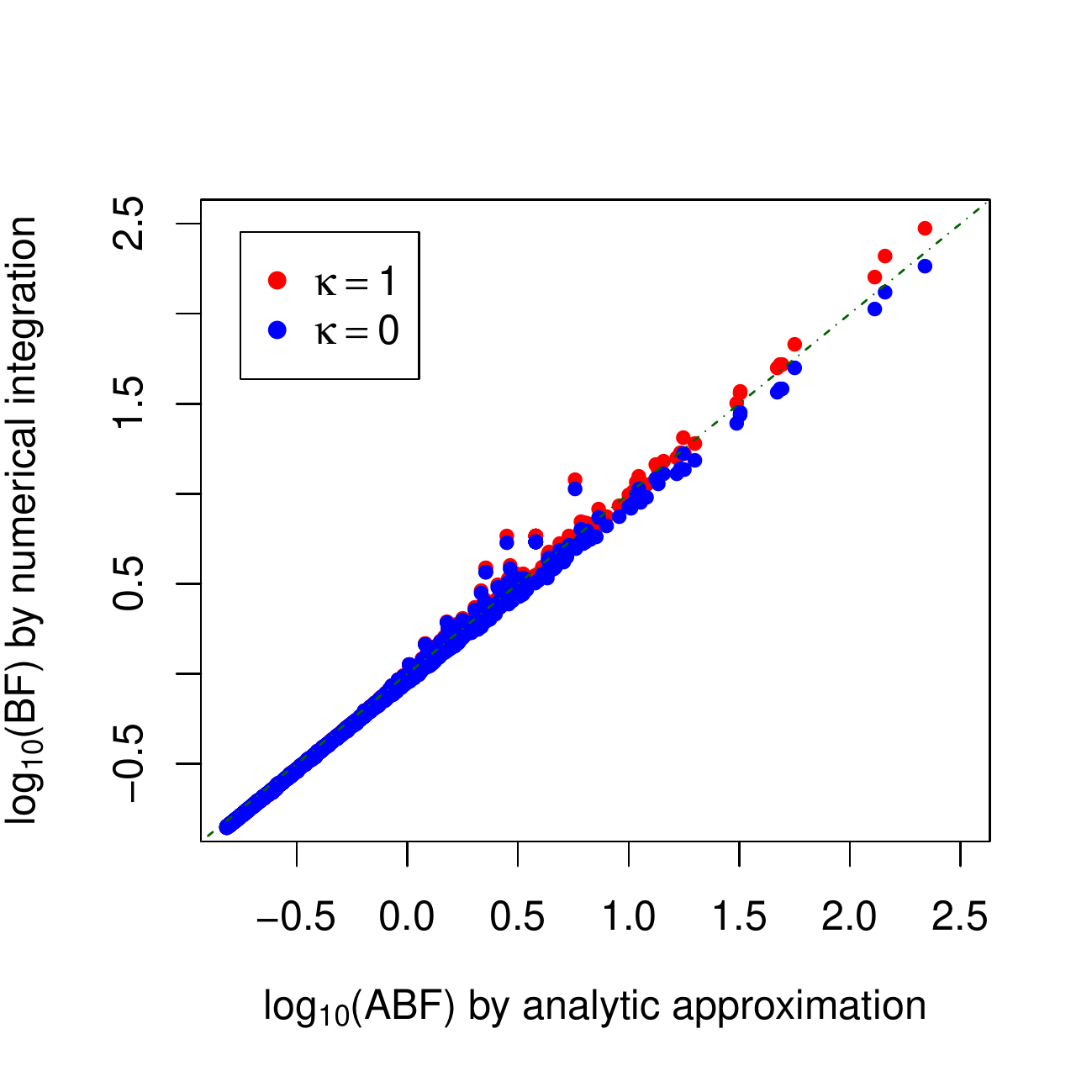}}
     \subfloat[sample size $n = 336$]{%
    \includegraphics[width=.40\textwidth]{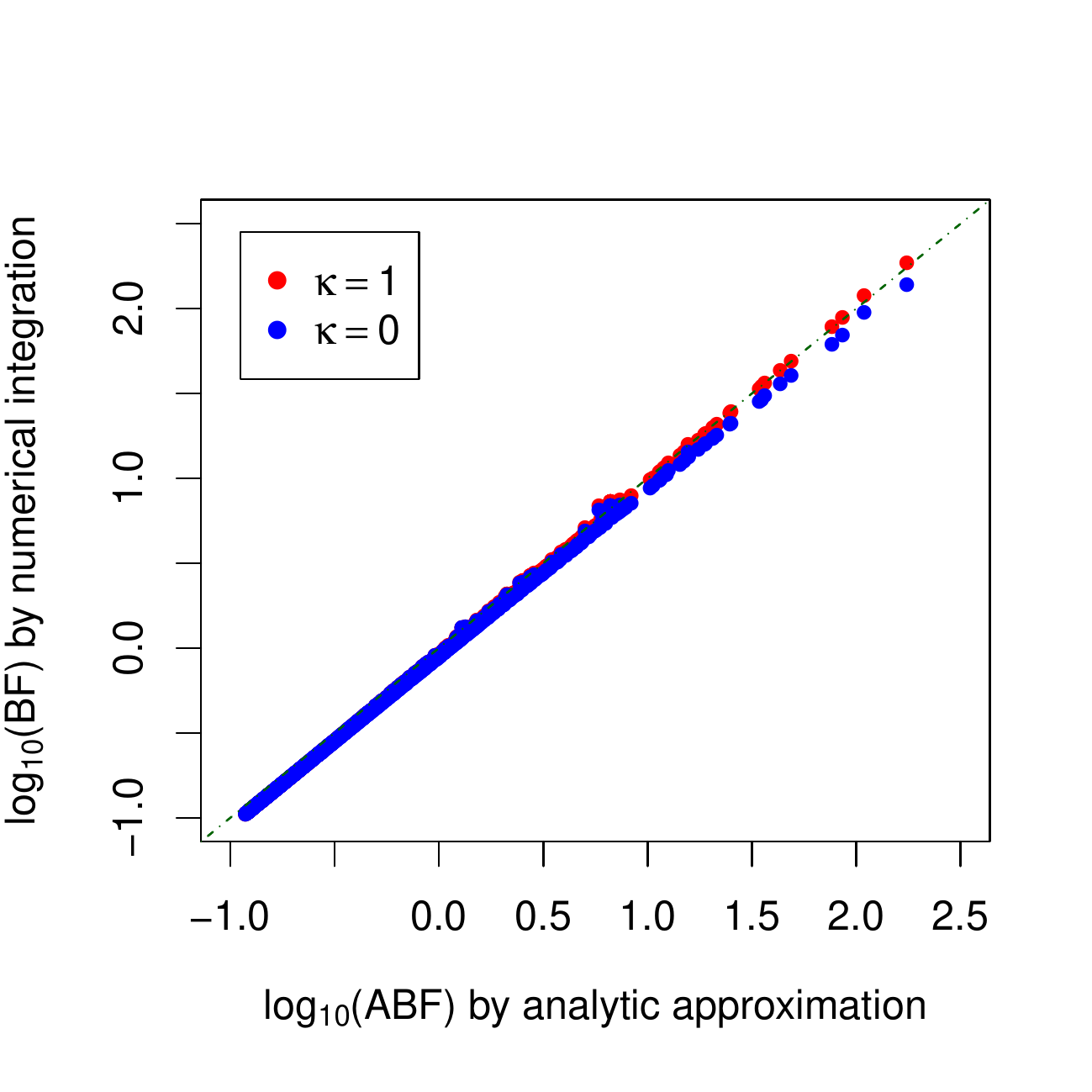}}\\
 \caption{\label{numerical.comp.fig} approximation accuracy by different sample sizes. In each panel, we plot the ``true" values of $\log_{10} \BF$ evaluated by numerical integration against their analytic approximations with different $\kappa$ values. For small sample sizes, $\ABF(\kappa=1)$ tends to be anti-conservative while $\ABF(\kappa=0)$ tends to be conservative comparing to the truth. Nevertheless, as the sample size grows, both become accurate.  
}  
\end{figure}

\section{Connection between Bayes factor and score statistic}

\subsection{Connection with fixed effect score statistic}

In this section, we give the mathematical details on connections between the approximate Bayes factor evaluated at $\kappa=0$ and the fixed effect score test statistics. In particular, it is sufficient to show that the quadratic form $\hat \beta(\tilde \lambda) \hat \Vv(\tilde \lambda, \tilde \tau)^{-1} \hat \beta(\tilde \lambda)$ corresponds to the score statistic for testing the fixed effect $\bv = 0$.

To see this, we relate $\hat \bv(\lambda)$ to $\tilde \av(\lambda)$, the MLE of $\av$ estimated under the null model restriction $\bv = 0$. The expression of $\tilde \av(\lambda)$ is given in (\ref{av.exp}), and it can be shown that
\begin{equation}
  \hat \bv (\lambda)  = \Qv(\lambda) \Gv'\Sigv(\lambda)^{-1}\left[\yv - \Xv \tilde \av(\lambda)\right],
\end{equation}  
where 
\begin{equation}
	 \Qv(\lambda) = \left[\Gv'\Sigv(\lambda)^{-1} \Gv -  \Gv'\Sigv(\lambda)^{-1} \Xv (\Xv'\Sigv(\lambda)^{-1} \Xv)^{-1} \Xv'\Sigv(\lambda)^{-1} \Gv\right]^{-1}.
\end{equation}
Furthermore,
\begin{equation}
   \hat \Vv (\lambda, \tau) = \tau^{-1} \Qv(\lambda)^{-1}.
\end{equation}
Therefore, it follows that
\begin{equation}\label{score.int.rst}
  \hat \Vv^{-1}(\lambda, \tau) \hbvg(\lambda) = \tau\Gv'\Sigv(\lambda)^{-1}\left(\yv - \Xv \tilde \av(\lambda) \right).
\end{equation}  
For $\lambda = \tilde \lambda$ and $\tau = \tilde \tau$ and noting the notations 
\begin{equation} \label{k0.def}
\begin{aligned}
 &  \tilde \av  = \tilde \av(\tilde \lambda), \\
 & \tilde \Sigv = \tilde \Sigma(\tilde \lambda), \\
 & \tilde \Qv = \tilde \Qv(\tilde \lambda), 
\end{aligned}
\end{equation}
the desired quadratic form $\hat \beta(\tilde \lambda) \hat \Vv(\tilde \lambda, \tilde \tau)^{-1} \hat \beta(\tilde \lambda)$ can be equivalent represented by
\begin{equation} \label{fix.score}
  \tilde \tau (\yv - \Xv \tilde \av)' \left[ \tilde \Sigv^{-1} \Gv \tilde \Qv \Gv' \tilde \Sigv^{-1}\right] (\yv - \Xv \tilde \av).
\end{equation}
It can be trivially derived from the first principle to show that expression (\ref{fix.score}) is indeed the score statistic under the LMM ((2.1) in main text) for testing the fixed effect $\bv = 0$.  

Alternatively,  we denote a projection matrix
\begin{equation}
  \tilde P_{x} = \Iv - \tilde \Sigv^{-1/2} \Xv (\Xv' \tilde \Sigv^{-1} \Xv)^{-1} \Xv'  \tilde \Sigv^{-1/2},
\end{equation}
and define 
\begin{equation}
  \tilde \Xv_G = \Gv' \tilde \Sigv^{-1/2} \tilde P_{x}. 
\end{equation}
We can further re-write (\ref{fix.score}) by
\begin{equation}
   \tilde  \tau \left((\yv - \Xv \tilde \av)'  \tilde \Sigv^{-1/2} \right) \left[\tilde \Xv_G (\tilde \Xv_G' \tilde \Xv_G)^{-1} \tilde \Xv_G' \right]  \left(  \tilde \Sigv^{-1/2}  (\yv - \Xv \tilde \av) \right).
\end{equation}
Note, matrix $\left[\tilde \Xv_G (\tilde \Xv_G' \tilde \Xv_G)^{-1} \tilde \Xv_G' \right]$ is also a projection matrix, and the results by \cite{Chen1983} indicate the above expression is indeed the desired score statistic.
 
\subsection{Connection with variance component score statistic}

The derivation of the score statistic based on LMM can be found in \cite{Chen2013}. 
Consider $\kappa=0$ and write $\Wv = \gamma \Mv$.  Under certain convergence condition (which typically requires the magnitude of $\gamma$ is bounded), it follows from the Neumann series expansion that
$$   \left(\Iv + \gamma \hat \Vv(\tilde \tau, \tilde \lambda)^{-1} \Mv\right)^{-1} = \sum_{n=0}^\infty (-\gamma)^n \left(\hat \Vv(\tilde \tau, \tilde \lambda)^{-1} \Mv\right)^n. $$ 
Combining the above expression with (\ref{score.int.rst}), 
\begin{equation} 
\begin{aligned} 
 &~~~\hbvg(\tilde \lambda)'\hat \Vv(\tilde \tau, \tilde \lambda)^{-1} (\gamma \Mv) \left[ \Iv +  \hat \Vv(\tilde \tau, \tilde \lambda)^{-1}(\gamma \Mv)\right]^{-1} \Vv(\tilde \tau, \tilde \lambda)^{-1} \hbvg(\tilde \lambda) \\
 & = \gamma \tilde \tau^2 \left(\yv - \Xv \tilde \av \right)'\tilde \Sigv^{-1}\Gv \Mv \Gv' \tilde \Sigv^{-1}\left(\yv - \Xv \tilde \av \right)  \\
 ~~~~~~~~~&+ \sum_{n=1}^\infty (-\gamma)^{n+1} \hbvg(\tilde \lambda)'\hat \Vv(\tilde \tau, \tilde \lambda)^{-1}  \Mv \left(\hat \Vv(\tilde \tau, \tilde \lambda)^{-1} \Mv\right)^n \Vv(\tilde \tau, \tilde \lambda)^{-1} \hbvg(\tilde \lambda) \\
 & = \gamma T_{\rm score} + \gamma^2 \sum_{n=0}^\infty (-\gamma)^{n} \hbvg(\tilde \lambda)'\hat \Vv(\tilde \tau, \tilde \lambda)^{-1}  \Mv \left(\hat \Vv(\tilde \tau, \tilde \lambda)^{-1} \Mv\right)^{n+1} \Vv(\tilde \tau, \tilde \lambda)^{-1} \hbvg(\tilde \lambda)
\end{aligned}   
 \end{equation}
 As $\gamma \to 0$, it follows that
 \begin{equation}
  \left(\Iv + \gamma \hat \Vv(\tilde \tau, \tilde \lambda)^{-1} \Mv\right)^{-1} = \Iv - \gamma \hat \Vv(\tilde \tau, \tilde \lambda)^{-1} \Mv + O(\gamma^2),
\end{equation}  
and
\begin{equation}
  \ABF(\Wv =  \gamma \Mv, \kappa=0) = \exp \left( \frac{\gamma}{2}\, T_{\rm score} \right)\cdot \big(1 + O(\gamma)\big).
\end{equation}

\section{MCMC algorithm for variable selection in BLMM} 

\cite{Wen2014a} provided an efficient MCMC algorithm to perform Bayesian variable selection in a very general complex linear model system. 
In the special case of a multiple linear regression model, their model selection formulation is almost identical to what we have described in section 4.2 of the main text, except when computing (approximate) Bayes factors, \cite{Wen2014a} assumes i.i.d residual errors and considers no random effect.   

Using the notations of (\ref{k0.def}), we note the BLMM induces a standard multiple linear  regression model on the transformed response variable, $\utilde{\yv} = \Sigv^{-\frac{1}{2}}\yv$, and transformed covariates, $\utilde{\Xv} = \Sigv^{-\frac{1}{2}}\Xv, \utilde{\Gv} = \Sigv^{-\frac{1}{2}}\Gv $. In particular, it is easy to see that, for arbitrary $\xi(\bv)$,  the approximate Bayes factors evaluated at $\kappa=0$ have identical values using either the original data $(\yv, \Xv, \Gv)$ or the transformed data $(\utilde{\yv}, \utilde{\Xv}, \utilde{\Gv})$, substituting $\tilde \Sigv$ for $\Sigv$. However, the induced linear model of the transformed data satisfies the requirement by the Metropolis-Hastings (M-H) algorithm described in \cite{Wen2014a}, i.e., i.i.d residual errors and no random effects. 

In practice, we implement the following algorithm to perform variable selections in the BLMM using the approximate Bayes factors for $\kappa=0$.
\begin{enumerate}
  \item Fit the null model and obtain $\tilde \lambda$.
  \item Compute $\tilde \Sigv$ and transform the observed data $(\yv, \Xv, \Gv)$ to $(\utilde{\yv}, \utilde{\Xv}, \utilde{\Gv})$.
  \item Apply the M-H algorithm of \cite{Wen2014a} on the transformed data set.
\end{enumerate}

The first step can be achieved by applying the software packages EMMAX \citep{Kang2010} or GEMMA \citep{Zhou2012}. The third step is implemented in the software package SBAMS \citep{Wen2014a}. To ensure $\ABF(\kappa=0)$ is faithfully computed, it is required to set ``-abf 0" option in SBAMS to estimated and use $\tilde \tau$ for all values of $\xi(\bv)$.

 \section{SNP set Simulation}
 
 In this section, we give the details of simulation schemes and parameters settings used in analysis.
 
\subsection{Simulation Details} 
 Our simulation scheme closely follows what is described in \cite{Lee2012}. For each simulated data set, we consider 5,000 non-overlapping SNP sets, with 1,000 SNPs in each set. 
For each SNP set, we simulate the genotypes of 2,000 individuals from a calibrated coalescent model \citep{Schaffner2005}, and the resulting LD structure within each SNP set mimics the LD patterns observed in European ancestry samples.   
For 3,500 out of 5,000 SNP sets, we simulate phenotypes from the null linear model
\begin{equation}
   \yv = 0.5 \xv + \ev,~ \ev \sim {\rm N}(0, \Iv),
\end{equation}
where $\xv$ is a vector of a continuous covariate generated from ${\rm N}(0,1)$. 
The phenotypes of the remaining 1,500 SNP sets are simulated from the following linear model,
\begin{equation}
   \yv = 0.5 \xv + \sum_i \beta_i \gv_i + \ev,~ \ev \sim {\rm N}(0, \Iv),
\end{equation}
by two different schemes. 
Both schemes assume 20\% of the SNPs in each SNP set are causal, and following \cite{Lee2012}, the effect size of each causal SNP is generated from the function $\beta_j = c | \log_{10}(m_j) |$, where $m_j$ is the minor allele frequency (MAF) of the assumed casual SNP. 
The two schemes differ in the choice of the parameter $c$. The first scheme assumes that all causal mutations are consistently deleterious by setting $c=0.1$ as a constant, and the second scheme assumes 40\% of the causal effects are protective by setting $c = -0.1$ and the rest of the 60\% causal mutations are deleterious by setting $c=0.1$. 
We vary the frequency of the sign-consistent alternative model ($\pi$) from 0.20 to 0.80. For each $\pi$ value, we simulate 10 independent data sets. 

\subsection{Analysis Details}
Following \cite{Lee2012}, we assume the same marginal weight vector, $\wv$, for SNPs in a SNP set in both SKAT and burden models. More specifically for each SNP $j$, we assign $w_j = {\rm Beta}({\rm MAF}_j,1,25)$, where ${\rm MAF}_j$ denotes the minor allele frequency of the $j$-th SNP.
Additionally for Bayesian analysis, we explicitly account for $\pi$ in our analysis, i.e., conditional on $H_0$ is false, we assume that with probabilities $\pi$ and $1-\pi$ the data are generated from the burden model and SKAT model, respectively.
We further re-normalize the marginal weights such that $\sum_j w_j =1$ and construct the priors on the standardized effect scale. 
More specifically, for the burden model, we assume $\Wv_b = \tau^{-1} \phi^2 (\sqrt{\wv}) (\sqrt{\wv})'$, and for the SKAT model, $\Wv_s = \tau^{-1} \phi^2 {\rm diag}(\wv)$. 
Under this normalized weighting formulation, the value of $\phi$ measures the prior overall magnitude of signal-noise ratio at the {\em set level}.
When computing the Bayes factors for both the burden and SKAT models, we consider a grid of $\phi$ values uniformly drawn from the set $\{\phi: 0.1, 0.2, 0.4, 0.8, 1.6 \}$.

We apply two different strategies in choosing the prior weights for computing Bayes factors. The first strategy simply assumes $\pi = \frac{1}{2}$. 
Alternatively, we estimate $\pi$ using an EM algorithm implemented in \citep{Flutre2013} by pooling all the SNP sets together. 
We then directly use the resulting Bayes factors and apply the FDR control procedure described in \citep{Wen2013} to perform hypothesis testing. 
For comparison, we apply the SKAT-O procedure \citep{Lee2012} and obtain the optimal $p$-value for each gene; then, we apply the Storey procedures to control FDR.

\section{Incorporating Detection of Common Variant Associations in SNP Set Testing}

Historically, the development of statistical methodology for SNP set testing is mostly motivated by detection of  rare genetic variant associations. 
Most recently, the study of expression quantitative trait loci (eQTLs) prompts integrating the ability of detecting both common and rare variant association signals into a unified SNP set analysis. 
More specifically, in eQTL analysis, the first line of the questions is to identify genes whose expression levels are regulated by either common or rare genetic variants (such genes are commonly referred to as eGenes). 
Most eGene detection approaches focus on the genomic region near the transcription start site of each target gene, i.e., the {\it cis} regulatory region. The SNPs within each region naturally form a candidate SNP set for each target gene.    

In the past, most eQTL studies have limited sample sizes and generally offer little power to detect rare variant associations.
As a consequence, most statistical methods for eGene detection, frequentist or Bayesian, aim to identify genes harbor common variant associations. 
But this is changing, many ongoing studies, e.g. the ongoing NIH GTEx project (\url{http://commonfund.nih.gov/GTEx/}), start collecting data with the sample size capable of discovering rare variant associations. 
In this section, we demonstrate that our Bayesian model averaging framework can efficiently combine existing approaches for testing both rare and common variant associations in SNP set testing for more powerful eGene discovery. 

\subsection{Bayesian Testing of Common Variant Associations in SNP Set}\label{common.set}

The Bayesian approaches for testing common variant associations in SNP sets have been proposed and applied in \cite{Servin2007, Flutre2013}. Here we extend their results into the context of BLMM.

Under the formulation of BLMM, to test against the null hypothesis, $H_0: \bv = 0$, we consider a specific class of alternative scenarios: exactly one variant in a set of $p$ candidate SNPs is truly associated \citep{Servin2007, Flutre2013}. Consequently,  there are $p$ different alternative models in total for a given SNP set.
Based on this simplifying assumption, without further information to distinguish the SNPs, we assign a discrete uniform prior to each SNP as {\em the} associated SNP.  We average over all $p$ alternative models and obtain the following Bayes factor for the SNP set
\begin{equation}
   \BF_{\rm cv} = \frac{1}{p} \sum_{i=1}^p \BF(\Wv_i),~\Wv_i = \phi^2 {\rm diag}(\gav_i)
\end{equation}     
where $\gav_i$ represents a $p$-dimensional binary indicator vector with only the $i$-th entry setting to 1. 
Clearly, each sepcification of $\Wv_i$ corresponds to an alternative single SNP association model, and each $\BF(\Wv_i)$ can be simplified to the form of equation (3.7) in the main text.   
 
It should be noted that the above alternative modeling approach is effective in identifying SNP set harboring common variant associations as demonstrated by \cite{Servin2007}, but it has very little power for rare variant testing.
It has been shown \citep{Wen2013} that the above Bayesian approach  has similar power comparing to the frequentist approach that takes the minimum single SNP association $p$-value as the test statistic for SNP set association. 

\subsection{Combined Bayesian SNP Set Testing of Common and Rare Variant Associations}

Based on the discussion in section \ref{common.set}, it is straightforward to formulate a Bayesian SNP set testing by averaging over three types of alternative models: burden, SKAT and the common variant (CV) models. A SNP set Bayes factor can be computed by
\begin{equation}\label{set.bf}
  \overline{\BF} = \pi_b \, \BF_{\rm burden} + \pi_s \, \BF_{\rm skat} + \pi_c \, \BF_{\rm cv},
\end{equation} 
where $\pi_b, \pi_s$ and $\pi_c$ denote the relative frequency of burden, SKAT and CV models in alternative settings. Although the default objective prior $\pi_b = \pi_s = \pi_c = \frac{1}{3}$ serves as a reasonable starting point, in typical eQTL studies, it is highly plausible to estimate these quantities using a hierarchical model by pooling information across genes genome-wide.  

In comparison, to the best of our knowledge, there is no existing frequentist SNP set testing approach that is optimally designed for detecting both common and rare variant associations in a computationally efficient way.  

\subsection{Extended Simulation Studies}

We perform additional simulation studies to demonstrate the power of the proposed Bayesian SNP set testing approach. Particularly, in addition to the schemes that simulate sign-consistent and sign-inconsistent multiple rare variant associations, we simulate a third type of alternative scenario where common variants drive the genetic associations. More specifically, we simulate according to the following linear model
\begin{equation}\label{cv}
  \yv = 0.5 \xv + \sum_{i \in S} \beta_i \gv_i + \ev, ~ \ev \sim {\rm N}(0, \Iv),
\end{equation}  
where $S$ denote a set of either one, two or three SNPs whose allele frequencies $\ge 0.05$. For each $i \in S$, we randomly draw $\beta_i$ from the distribution $ {\rm N}(0, 1)$. 
    
In each simulated data set, we still consider 5,000 SNP sets where 2,500 are simulated from the null model. For the remaining 2,500 SNP sets, we simulate the phenotypes using the three types of alternative models according to a pre-defined parameter $\piv = (\pi_1, \pi_2, \pi_3)$, where $\pi_1, \pi_2$ and $\pi_3$ denote the relative frequency of sign consistent model, sign inconsistent model and common variant model (\ref{cv}), respectively. We vary $\piv$ values to generate different simulated data sets.       
     
To analyze the simulated data set,  we estimate $(\pi_b, \pi_s, \pi_c)$ using the same hierarchical model employed in the rare variants SNP set testing and compute the Bayes factor based on (\ref{set.bf}).  For  comparison, we compute the SKAT-O $p$-values for each SNP set. We again examine the FDR control and power for each analysis method.

The results from those additional simulation studies are summarized in Table \ref{additional.sim.tbl}. 
All experimented methods control FDR at the desired level. 
Interestingly, we find both burden and SKAT models have decent power in detecting SNP sets harboring common variant associations, which is reflected by the overall good performance of the SKAT-O approach. 
Nevertheless as expected,  by explicitly targeting  and modeling all possible alternative scenarios, the Bayesian approach yield substantially higher power.

\begin{table}[h!t]
\begin{center}
{\fontsize{10}{12}\selectfont 
\begin{tabular} { c  c    c c  c  c  c }
\hline
   ~ & ~& \multicolumn{2}{c}{FDR} &~& \multicolumn{2}{c}{Power}\\
\cline{3-4} \cline{6-7}
 Setting ($\piv$) & ~& SKAT-O  & Bayesian-E & ~ & SKAT-O  & Bayesian-E \\
\hline
\hline
 $(0.25, 0.35, 0.40)$ & ~ &  0.046 & 0.044 & ~ & 0.831 & 0.921 \\
 $(0.25,  0.25, 0.50)$ & ~ &  0.041 & 0.042 & ~ & 0.815  & 0.917 \\
 $(0.25, 0.15, 0.60)$ &~&  0.050 & 0.049 & ~ & 0.799 & 0.904 \\
 $(0.20, 0.10, 0.70)$ &~& 0.047 & 0.048 & ~  & 0.789&  0.873 \\
 $(0.10, 0.15, 0.75)$ &~& 0.042 & 0.042 & ~ & 0.757 &  0.872 \\
\hline
\end{tabular}}
\caption{ \small \label{additional.sim.tbl} Realized false discovery rate and power in extended simulation studies of SNP set analysis. The first column (setting) indicates the distribution of SNP sets simulated by sign-consistent, sign-inconsistent, and common variant association models. For the SKAT-O procedure, the resulting $p$-values are further processed by the Storey procedure for FDR controls.  The Bayesian approach ("Bayesian-E")  estimates $\pi$ from the data and compute the Bayes factor for each SNP set according to (\ref{set.bf}). The FDR control for the Bayes factors is performed using the EBF procedure described in \cite{Wen2013}.  }

\end{center}
\end{table}

\bibliographystyle{natbib}  
\bibliography{BLMM,BLMM_supp}

\begin{thebibliography}{}

\bibitem[Baxter {\em et~al.}(2010)Baxter, Brazelton, Yu, {\em
  et~al.}]{Baxter2010}
Baxter, I., Brazelton, J.~N., Yu, D., {\em et~al.} (2010).
\newblock A coastal cline in sodium accumulation in arabidopsis thaliana is
  driven by natural variation of the sodium transporter athkt1; 1.
\newblock {\em PLoS genetics\/}, {\bf 6}(11), e1001193.

\bibitem[Brooks {\em et~al.}(2003)Brooks, Giudici, and Philippe]{Brooks2003}
Brooks, S., Giudici, P., and Philippe, A. (2003).
\newblock Nonparametric convergence assessment for mcmc model selection.
\newblock {\em Journal of Computational and Graphical Statistics\/}, {\bf
  12}(1), 1--22.

\bibitem[Butler(2007)Butler]{Butler2007}
Butler, R.~W. (2007).
\newblock {\em Saddlepoint approximations with applications\/}.
\newblock Cambridge University Press.

\bibitem[Chen(1983)Chen]{Chen1983}
Chen, C.-F. (1983).
\newblock Score tests for regression models.
\newblock {\em Journal of the American Statistical Association\/}, {\bf
  78}(381), 158--161.

\bibitem[Chen {\em et~al.}(2013)Chen, Meigs, and Dupuis]{Chen2013}
Chen, H., Meigs, J.~B., and Dupuis, J. (2013).
\newblock Sequence kernel association test for quantitative traits in family
  samples.
\newblock {\em Genetic Epidemiology\/}, {\bf 37}(2), 196--204.

\bibitem[Flutre {\em et~al.}(2013)Flutre, Wen, Pritchard, and
  Stephens]{Flutre2013}
Flutre, T., Wen, X., Pritchard, J., and Stephens, M. (2013).
\newblock A statistical framework for joint eqtl analysis in multiple tissues.
\newblock {\em PLoS genetics\/}, {\bf 9}(5), e1003486.

\bibitem[Good(1992)Good]{Good1992}
Good, I. (1992).
\newblock The bayes/non-bayes compromise: A brief review.
\newblock {\em Journal of the American Statistical Association\/}, {\bf
  87}(419), 597--606.

\bibitem[Guan {\em et~al.}(2011)Guan, Stephens, {\em et~al.}]{Guan2011}
Guan, Y., Stephens, M., {\em et~al.} (2011).
\newblock Bayesian variable selection regression for genome-wide association
  studies and other large-scale problems.
\newblock {\em The Annals of Applied Statistics\/}, {\bf 5}(3), 1780--1815.

\bibitem[Kang {\em et~al.}(2010)Kang, Sul, Service, {\em et~al.}]{Kang2010}
Kang, H.~M., Sul, J.~H., Service, S.~K., {\em et~al.} (2010).
\newblock Variance component model to account for sample structure in
  genome-wide association studies.
\newblock {\em Nature genetics\/}, {\bf 42}(4), 348--354.

\bibitem[Lee {\em et~al.}(2012)Lee, Wu, and Lin]{Lee2012}
Lee, S., Wu, M.~C., and Lin, X. (2012).
\newblock Optimal tests for rare variant effects in sequencing association
  studies.
\newblock {\em Biostatistics\/}, {\bf 13}(4), 762--775.

\bibitem[Madsen and Browning(2009)Madsen and Browning]{Madsen2009}
Madsen, B.~E. and Browning, S.~R. (2009).
\newblock A groupwise association test for rare mutations using a weighted sum
  statistic.
\newblock {\em PLoS genetics\/}, {\bf 5}(2), e1000384.

\bibitem[Neale {\em et~al.}(2011)Neale, Rivas, Voight, {\em et~al.}]{Neale2011}
Neale, B.~M., Rivas, M.~A., Voight, B.~F., {\em et~al.} (2011).
\newblock Testing for an unusual distribution of rare variants.
\newblock {\em PLoS genetics\/}, {\bf 7}(3), e1001322.

\bibitem[Schaffner {\em et~al.}(2005)Schaffner, Foo, Gabriel, {\em
  et~al.}]{Schaffner2005}
Schaffner, S.~F., Foo, C., Gabriel, S., {\em et~al.} (2005).
\newblock Calibrating a coalescent simulation of human genome sequence
  variation.
\newblock {\em Genome research\/}, {\bf 15}(11), 1576--1583.

\bibitem[Schifano {\em et~al.}(2012)Schifano, Epstein, Bielak, {\em
  et~al.}]{Schifano2012}
Schifano, E.~D., Epstein, M.~P., Bielak, L.~F., {\em et~al.} (2012).
\newblock Snp set association analysis for familial data.
\newblock {\em Genetic epidemiology\/}, {\bf 36}(8), 797--810.

\bibitem[Segura {\em et~al.}(2012)Segura, Vilhj{\'a}lmsson, Platt, {\em
  et~al.}]{Segura2012}
Segura, V., Vilhj{\'a}lmsson, B.~J., Platt, A., {\em et~al.} (2012).
\newblock An efficient multi-locus mixed-model approach for genome-wide
  association studies in structured populations.
\newblock {\em Nature genetics\/}, {\bf 44}(7), 825--830.

\bibitem[Servin and Stephens(2007)Servin and Stephens]{Servin2007}
Servin, B. and Stephens, M. (2007).
\newblock Imputation-based analysis of association studies: candidate regions
  and quantitative traits.
\newblock {\em PLoS genetics\/}, {\bf 3}(7), e114.

\bibitem[Stephens and Balding(2009)Stephens and Balding]{Balding2009}
Stephens, M. and Balding, D.~J. (2009).
\newblock Bayesian statistical methods for genetic association studies.
\newblock {\em Nature Reviews Genetics\/}, {\bf 10}(10), 681--690.

\bibitem[Wakefield(2009)Wakefield]{Wakefield2009}
Wakefield, J. (2009).
\newblock Bayes factors for genome-wide association studies: comparison with
  p-values.
\newblock {\em Genetic epidemiology\/}, {\bf 33}(1), 79--86.

\bibitem[Wen(2013)Wen]{Wen2013}
Wen, X. (2013).
\newblock Robust bayesian fdr control with bayes factors.
\newblock {\em arXiv preprint arXiv:1311.3981\/}.

\bibitem[Wen(2014)Wen]{Wen2014a}
Wen, X. (2014).
\newblock Bayesian model selection in complex linear systems, as illustrated in
  genetic association studies.
\newblock {\em Biometrics\/}, {\bf 70}(1), 73--83.

\bibitem[Wen and Stephens(2014)Wen and Stephens]{Wen2014b}
Wen, X. and Stephens, M. (2014).
\newblock Bayesian methods for genetic association analysis with heterogeneous
  subgroups: from meta-analyses to gene-environment interactions.
\newblock {\em Annals of Applied Statistics\/}, {\bf 8}(1), 176--203.

\bibitem[Wu {\em et~al.}(2011)Wu, Lee, Cai, {\em et~al.}]{Wu2011}
Wu, M.~C., Lee, S., Cai, T., {\em et~al.} (2011).
\newblock Rare-variant association testing for sequencing data with the
  sequence kernel association test.
\newblock {\em The American Journal of Human Genetics\/}, {\bf 89}(1), 82--93.

\bibitem[Zhou and Stephens(2012)Zhou and Stephens]{Zhou2012}
Zhou, X. and Stephens, M. (2012).
\newblock Genome-wide efficient mixed-model analysis for association studies.
\newblock {\em Nature genetics\/}, {\bf 44}(7), 821--824.

\bibitem[Zhou {\em et~al.}(2013)Zhou, Carbonetto, and Stephens]{Zhou2013}
Zhou, X., Carbonetto, P., and Stephens, M. (2013).
\newblock Polygenic modeling with bayesian sparse linear mixed models.
\newblock {\em PLoS genetics\/}, {\bf 9}(2), e1003264.

\end{thebibliography}

\end{document}